\documentclass[12pt]{article}
\usepackage{epsfig, amssymb, graphicx, amsmath} 
\usepackage{amsthm}
\usepackage{multirow}
\usepackage{epstopdf} 	
\usepackage{multicol}	
\usepackage{booktabs}
\usepackage[shortlabels]{enumitem}

\usepackage{esvect}	
\usepackage{eurosym}	
\usepackage{wrapfig}	
\usepackage{float}		
\usepackage{subfig}	
\usepackage{footnote}
\usepackage{amsthm}
\usepackage{thmtools}
\usepackage[toc,page]{appendix}
\usepackage{algorithm}

\usepackage[noend]{algpseudocode}
\usepackage{bbm}		

\usepackage[authoryear]{natbib}
\bibpunct{(}{)}{,}{a}{}{,}

\newtheorem{theorem}{Theorem}

\newtheorem{proposition}[theorem]{Proposition}

\theoremstyle{definition}

\theoremstyle{remark}

\numberwithin{theorem}{section}
\numberwithin{proposition}{section}
\numberwithin{lemma}{section}
\numberwithin{corollary}{section}
\numberwithin{definition}{section}
\numberwithin{remark}{section}

\newcommand{\E}{\mathbb{E}}

\newcommand{\be}{\begin{equation}}
\newcommand{\en}{\end{equation}}
\newcommand{\ben}{\begin{equation*}}
\newcommand{\enn}{\end{equation*}}
\newcommand{\bea}{\begin{eqnarray}}
\newcommand{\ena}{\end{eqnarray}}
\newcommand{\pushright}[1]{\ifmeasuring@#1\else\omit\hfill$\displaystyle#1$\fi\ignorespaces}

\begin{document}
	
	\newlength\tindent
	\setlength{\tindent}{\parindent}
	\setlength{\parindent}{0pt}
	\renewcommand{\indent}{\hspace*{\tindent}}
	
	\begin{savenotes}
		\title{
			\bf{ 
				A fast Monte Carlo scheme for \\ additive processes and option pricing
		}}
		\author{
			Michele Azzone$^{\ddagger\P}$ \& 
			Roberto Baviera$^{\ddagger}$ 
		}
		
		\maketitle
		
		\vspace*{0.11truein}
		\begin{tabular}{ll}
			$(\ddagger)$ &  Politecnico di Milano, Department of Mathematics, 32 p.zza L. da Vinci, Milano \\
			$(\P)$ & European Central Bank\footnote{The views expressed are those of the author and do not necessarily reflect the views of ECB.}, 20 Sonnemannstraße, Frankfurt am Main, Germany\\
		\end{tabular}
	\end{savenotes}
	
	\vspace*{0.11truein}
	
	\begin{abstract}
		\noindent 
		In this paper, we present a very fast Monte Carlo scheme for additive processes: 	the computational time is of the same order of magnitude of standard algorithms for simulating Brownian motions.
		We analyze in detail numerical error sources and propose a technique that reduces the two major sources of error.
		We also compare our results with a benchmark method: the jump simulation with Gaussian approximation.
		
		\noindent
		We show an application to additive normal tempered stable processes,  a class of additive processes that calibrates ``exactly" the implied volatility surface.
		Numerical results are relevant. 
		This fast algorithm is also an accurate tool for pricing path-dependent discretely-monitoring options with errors of one basis point or below. 
	
	\end{abstract}
	
	\vspace*{0.11truein}
	\vspace*{0.11truein}
	
	\vspace{4cm}
	\begin{multicols}{2}
			{\bf Address for correspondence:}\\
		{\bf Prof. Roberto Baviera}\\
		Department of Mathematics \\
		Politecnico di Milano\\
		32 p.zza Leonardo da Vinci \\ 
		I-20133 Milano, Italy \\
		Tel. +39-02-2399 4575\\
		roberto.baviera@polimi.it
		\columnbreak
	
		$ $\\
		{\bf Dr. Michele Azzone}\\
		European Central Bank\\
		Sonnemannstraße 20 \\ 
		D-60314 Frankfurt am Main, Germany \\
		Tel. +39-338-2464 527 \\
		michele.azzone@ecb.europa.eu\\
	\end{multicols}

	\newpage
	
	\begin{center}
		\Large\bfseries 
		A fast Monte Carlo scheme for \\ additive processes and option pricing
	\end{center}
	
	\vspace*{0.21truein}
	
	\section{Introduction}
	
	In this paper, we introduce a fast Monte Carlo simulation technique for additive processes. 
	In option pricing, Monte Carlo methods are attractive because they do not require significant modifications when the payoff structure of the derivative
	changes.  We describe an efficient and accurate algorithm for Monte Carlo simulations of the process increments
	and we compute the prices of a class of discretely-monitoring path-dependent options. A process $\{X(t)\}_{t\ge0}$ is said to be an additive process, 
	if it presents independent (but not-stationary) increments and satisfies $X(0) = 0$ a.s.; stationarity is the main difference with L\'evy processes \citep[see e.g.,][]{Sato1999}. 
	
	\smallskip
	
	Additive processes are becoming the new frontier in equity derivatives 
	for their ability, 
	on the one hand, to reproduce accurately market data, and
	on the other hand,  to keep the process rather elementary \citep[see e.g.,][]{madan2020additive,carr2021additive,azzone2019additive}.
	In this paper, we show another advantage of additive processes: simulation schemes are as fast as standard (fast) algorithms for simulating the Black-Scholes model.

	
	\smallskip
	
	Up to our knowledge, the unique Monte Carlo (MC) scheme developed for a specific class of additive processes, Sato processes, 
	has been introduced by \citet{EberleinMadan2009}.
	They generalize to this class of  additive processes, a well-known jump simulation technique developed for L\'evy  processes, that can be found in many excellent
	textbooks   \citep[see e.g.,][]{Cont, Asmussen2007book}. 
	It entails truncating small jumps below a certain threshold and then simulating the finite number of independent jumps;  
	finally, the \cite{Asmussen2001approximations} Gaussian approximation (hereinafter GA) can be applied to substitute small jumps with a diffusive term:
	this has become a benchmark technique to compare numerical results.
	
	\smallskip
	
	In this paper, we propose a new MC technique for additive processes based on a numerical inversion of the cumulative distribution function (CDF).
	Monte Carlo simulation of additive processes is not straightforward because, in general, the CDF of process increments is not known explicitly.  However,
	analytic expressions exist for the characteristic functions thanks to the celebrated L\'evy--Khintchine formula \citep{Sato1999}.
	Since the seminal paper of  \citet{bohman1970method}, 
	general methods have been developed for sampling from Fourier transforms and
	even some specific methods for some distributions (e.g. stable distributions) that do not require numerical inversion \citep[][Sec.1.7, p.41]{Samorodnitsky1994}.
	
	In the financial literature,  these techniques have  been developed specifically in the L\'evy case, where it is possible to leverage on the stationary  increments
	\citep[see e.g.,][]{glasserman2010sensitivity,chen2012simulating,BallottaKyriakou2014}. 
	These techniques are reliable and efficient: 
	they build upon the characteristic function numerical inversion to obtain an estimation of the CDF. 
	Specifically, we use the fast Fourier transform (FFT) method for the numerical inversion as 
	proposed by \citet{Lee2004option} and then applied to MC option pricing
	in the studies of \citet{chen2012simulating} and \citet{BallottaKyriakou2014}. Unfortunately, it is not trivial to extend these numerical methods to additive processes.
	Relative to this literature, our contribution lies in \textit{i)} extending to non-stationary processes these techniques and \textit{ii)}
	analyzing the three sources of error that arise in estimating derivative price expectations and 
	showing how to improve the two largest ones.

	\smallskip
	
	Three are the main contributions of this paper. 
	First, we propose a new
	Monte Carlo simulation technique for additive processes based on FFT. 
	Second, we improve the two main sources of numerical error in existing techniques to accelerate convergence, 
	using both a property of the Lewis formula in the complex plane and a spline method for CDF numerical inversion. 
	Finally, we point out that the proposed technique is accurate and fast: $i$) we compare it with traditional GA simulations showing that it is at least one and a half orders of magnitude faster 
	whatever time horizon we consider and $ii$) we observe that, when pricing some  discretely-monitoring path-dependent options, the computing
	time has the same order of magnitude as standard algorithms for Brownian motions.
	
	\smallskip
	
	The rest of the paper is organized as follows.
	In section \ref{sec:Overview}, we overview the method and recall both \citet{Lewis2001} formula for CDF and the error source in the numerical approximation:
	we discuss the optimal selection of the integration path. 
	In section \ref{sec:SimulationMethod}, we describe the proposed simulation method and present the other main error source in MC option pricing: 
	the interpolation method in numerical inversion. We also discuss how to generalize the GA method for additives in an efficient way.
	Section \ref{sec:NumericalResults} presents numerical results for a large class of pure-jump additive processes in the case of 
	both European options (where analytic pricing methods are available), and some discretely-monitoring path-dependent options.
	Section \ref{sec:Conclusion} concludes. In a dedicated section, we report the notation and abbreviations used in the paper.
The main characteristics of the additive process that we use for the numerical analysis can be found in appendix \ref{appendix_ATS}, a brief description of the algorithm in appendix \ref{appendix:simulation_algo} and a comparison of simulated option prices with and without spline interpolation in appendix \ref{appendix:table}.

	\section{Overview of the MC method for additive process}
	\label{sec:Overview}
	
	Pure jump asset pricing models based on additive processes have 
	enjoyed remarkable popularity in recent years.
	At least for two main reasons. First, they allow a highly tractable closed-form approach with simple analytic expression for European options following \cite{Lewis2001}.
	This formula is computable as fast as the standard Black-Scholes one.
	Second, additive processes provide an adequate calibration to the implied volatility surface of equity derivatives, 
	as well as they reproduce {\it stylized facts} as the time scaling of skew in volatility smile \citep[see e.g.,][]{azzone2022short}.
	
	In this section, we describe a third reason in favor of these models: 
	they allow a simple, accurate, and fast numerical scheme for path-dependent option valuation.  
	We extend to additive processes the preceding literature on L\'evy processes' simulation techniques and we discover that, thanks to this Monte Carlo scheme, it is possible to price efficiently exotic derivatives as Asian contracts or barrier options with discretely-monitored barriers, because we can focus only on monitoring dates.
	
	\smallskip
	
	The simulation of a discrete sample path of an additive process reduces to simulating from the distribution of the process increment between time $s$ and time $t>s$.
	L\'evy process simulation is based on time-homogeneity of the jump process: the characteristic function of an increment is the same as the characteristic function of the process itself at time $t=1$, re-scaled by the time interval $(t-s)$ of interest.
	
	\smallskip
	
	In this paper, we extend the preceding analysis to L\'evy processes by  $i$) presenting an explicit method for additive processes  from their characteristic function and
	$ii$) analyzing the explicit bound for the total estimation bias.
	In the L\'evy case,
	thanks to process time-homogeneity,
	the properties of the process characteristic function are immediately extended to its increments.
	For example, the characteristic function (also of increments) 
	is analytic in a horizontal strip and the purely imaginary points on the boundary of the strip of
	regularity are singular points  
	\citep[cf.][th.3.1, p.12]{Lukacs1972}.
	This identification of process characteristic function and increments' characteristic function is not anymore valid for additive processes.
	However, the present paper shows that the analyticity strip depends on time and that it is possible to build an efficient numerical scheme for additive processes. Let us point out that  it is not trivial to extend to the non-stationary case of additive processes other advanced methods developed for simulating L\'evy processes \citep[see e.g.,][]{kuzntsov2011wiener,ferreiro2015applying,boyarchenko2019sinh,kudryavtsev2019approximate} or pricing path-dependent derivatives \citep[see e.g.,][]{jackson2008fourier,phelan2019hilbert}.
	
	\smallskip
	
	Our method is based on three key observations. 
	First, computing a CDF $P(x)$ corresponds to pricing a digital option: this can be done efficiently in the Fourier space. This step can be crucial, as already highlighted by \citet{BallottaKyriakou2014} in the L\'evy case, 
	the standard Fourier formula with Hilbert transform presents some numerical instabilities due to the presence of a pole in the origin. 
	They propose a regularization that leads to an additional numerical error.
	We propose a different approach that is based on the  \citet{Lewis2001}
	formula which presents two significant advantages. 
	On the one hand, this technique is exact (thus, no numerical error is associated with it), and, on the other hand, it allows selecting the optimal integration path that reduces the numerical error in the discretization of the CDF.
	
	Second,  the \citet{Lewis2001} formula for the CDF can be viewed as an inverse Fourier transform method that
	can be approximated with a fast Fourier transform (FFT) technique: Lewis-FFT computes multiple values
	of the CDF simultaneously in a very efficient way.
	
	Finally, knowing the CDF approximation $\hat{P}$, we can sample from this
	distribution by  inverting the CDF, i.e. by
	setting $X =\hat{P}^{-1} (U)$, with $U$ a uniform r.v. in $[0,1]$.
	Thus, simulating a r.v. via a numerical CDF (i.e. coupling the discrete Fourier
	transform with a Monte Carlo simulation), requires a numerical inversion that is realized via an interpolation method.
	Following \citet{glasserman2010sensitivity},  due to its simplicity, a linear interpolation of the CDF is chosen in the existing financial literature \citep[see e.g.,][]{chen2012simulating, FengLin2013}.
	We propose the spline as interpolation rule because the computational cost is  similar, 
	while the bias associated with the  two interpolation rules is significantly different:
	the upper bound of the bias can be estimated  for a given grid step $\gamma$, 
	and, as we discuss in section \ref{sec:SimulationMethod}, it should be at least $\gamma^2$ smaller for the spline interpolation.
	In extensive numerical experiments we observe that, on the one hand, the error decreases even faster as a power of $\gamma$ than predicted by the upper bound, 
	thanks to the additional properties of the interpolated functions, 
	and on the other hand, it becomes negligible for the grids that are selected in practice.
	
	Due to these three main ingredients (Lewis formula, FFT and Spline interpolation) that play a crucial role in the proposed Monte Carlo simulation technique, 
	we call the method Lewis-FFT-S. The algorithm is reported in appendix \ref{appendix:simulation_algo}. 
	
	\smallskip
	
	The Lewis-FFT-S method extends the \citet{EberleinMadan2009} technique to any additive process of financial interest, being significantly faster: 
	we show that the proposed Monte Carlo  is much faster than any jump-simulation method even considering the
	\citet{Asmussen2001approximations}
	Gaussian approximation.
	Analyzing in detail the numerical errors related to the methodology,
	we design an algorithm that increases both accuracy and computational efficiency.
	To the best of our knowledge, the proposed scheme is the first application in financial engineering of the MC simulation based on Lewis formula and FFT, when the underlying is governed by an additive process.
	
	\smallskip
	
	In the next subsection, we also recall explicit and computable expressions for the error estimates.

	\subsection{Lewis CDF via FFT}

	The proposed MC method simulates from the characteristic function of the additive increments.
	Due to the L\'evy-Khintchine formula, the characteristic function 
	\[
	\phi_t (u) := \E \, e^{i \; u \; f_t }
	\]
	of an additive process $f_t$  admits a closed-form expression. 
	Furthermore, as already mentioned, 
	according to \citet[][th.3.1, p.12]{Lukacs1972}, the process characteristic function 
	is analytic in a horizontal strip of the complex plane. 
	Similarly to \citet{Lee2004option},
	we define $p^-_t\geq0$ and $-(p^+_t +1) \leq 0$,  s.t. the characteristic function $\phi_t$
	is analytic  when $ \Im (u) \in (-(p^+_t + 1), p^-_t)$. 
	
	We observe that for Levy processes, the increment $f_{t} - f_{s}$ has the same distribution as $f_{\Delta}$, where 
	$\Delta=t-s$: the same property does not hold for additive processes, due to the time inhomogeneity.
	For an additive process, the characteristic function of an increment $f_t - f_s$ between times $s$ and $t>s$ is 
	\[
	\phi_{s,t}  (u) = \E \, e^{i \; u \; (f_t - f_s) } = \displaystyle \dfrac{ \E \, e^{i \; u \; f_t }}{  \E \, e^{i \; u \; f_s }} \,  ,
	\]
	due to the independent increment property of additive processes.

	Moreover, for all additive processes, a relevant property holds on the analytic strip of the characteristic function $	\phi_t $.
	\bigskip
	
	\begin{theorem}\label{prop:assump_1}
		$p_t^+$ and $p_t^-$ are non increasing for all additive processes.
	\end{theorem}
	\begin{proof}
		From theorem 9.8 of \citet[][p.52]{ Sato1999}, we have that for any additive process the L\'evy measure $\nu_t(x)$ is a positive and non decreasing function of $t$ for any $x$. 
	Thanks to the L\'evy Khintchine  representation the characteristic exponent of an additive process, given its triplet $(\gamma_t,\,A_t,\,\nu_t)$, is \citep[see e.g.,][th.8.1 p.37]{Sato1999}
	\begin{equation}
		\log\phi_t=iu\gamma_t-u^2A_t+\int_\mathbb{R} dx\,(e^{iux}-1-I_{|x|<1}i\,u\,x)\nu_t(x)\;\;, \label{eq:char_jump}
	\end{equation}
	where $\gamma_t$ is the drift term and $A_t$ the diffusion term.	\citet[][th.3.1, p.12]{Lukacs1972} has proven that the characteristic function	is analytical in an horizontal strip that includes the origin and is delimited by two points (if the strip is not the whole plane) on the imaginary axis. Hence,  we evaluate the characteristic function in $u=-i\,a$, with $a\in \mathbb{R}$, and identify $p_t^+$ and $p_t^-$ as the extrema of the interval of $a$ s.t. (\ref{eq:char_jump}) is well defined, i.e. $a\in (-p_t^-,p_t^++1)$.  The integral $\int_\mathbb{R}dx\, (e^{ax}-1-I_{|x|<1}\,a\,x)\nu_t(x)$ is the unique term that can diverge in (\ref{eq:char_jump}).\\  First, we recall that $\nu_t(x)$ is bounded for $x\neq 0$ and $\int_{\mathbb{R}}dx\,\min(|x|^2,1)\nu_t(x)<\infty$ \citep[cf.][th.8.1 p.37]{Sato1999}. Then, for any $Q>1$  the quantities $i)$ $\int_{-Q}^Qdx\, (e^{ax}-1-I_{|x|<1}\,a\,x)\nu_t(x)$,  $ii)$ $\int_{-\infty}^{-Q} dx\,\nu_t(x)$ and $iii)$ $\int_{Q}^{\infty}dx\, \nu_t(x)$ are finite. 
	Thus, we can recognize $p_t^+$ and $p_t^-$ from the set of $a$ for which  $ \int_{Q}^\infty dx\, e^{ax}\nu_t(x)$ and $ \int_{-\infty}^{-Q} dx\, e^{ax}\nu_t(x)$  converge. 
	\smallskip\\
	Let us first prove the proposition for $p_t^+$. Notice that $p_t^+$ is unique because $\int_{Q}^\infty dx\, e^{ax}\nu_t(x)$ is non decreasing in $a$ and that $p_t^+\geq -1$ because the origin is included in the analytical strip. Fix $t>0$, there are three possible cases
	\begin{enumerate}
		\item if $ \int_{Q}^\infty dx\, e^{ax}\nu_t(x)=\infty$ for any $a>0$, then $p_t^+=-1$;
		\item if $ \int_{Q}^\infty dx\,e^{ax}\nu_t(x)<\infty$ for any $a>0$, then $p_t^+=\infty$;
		\item if it exists $\lambda_t^+$ s.t. $ \int_{Q}^\infty\, e^{ax}\nu_t(x)dx<\infty$ for any $0<a<\lambda_t^+$ and $ \int_{Q}^\infty dx\, e^{ax}\nu_t(x)=\infty$ for any $a>\lambda_t^+$, then $p_t^+=\lambda_t^+-1$.
	\end{enumerate}
	For any $s<t$,  we observe that $\int_{Q}^\infty dx\, e^{ax}\nu_s(x) \leq \int_{Q}^\infty dx \,e^{ax}\nu_t(x)$, thanks to the monotonicity of $\nu_t(x)$ in $t$: let us consider the implications on the monotonicity of $p_t^+$ in the three cases.\\ In case 1, $p_t^+\leq p_s^+$ because $p_s^+\geq -1$ as emphasized above. In case 2, $ \int_{Q}^\infty dx\, e^{ax}\nu_s(x)\leq \int_{Q}^\infty dx\, e^{ax}\nu_t(x)<\infty$ and then $p_s^+=\infty$. In case 3, also $ \int_{Q}^\infty dx \,e^{ax}\nu_s(x)<\infty$ for any $0<a<\lambda_t^+$ and then $\lambda_t^+\leq \lambda_s^+$, i.e. $p_t^+\leq p_s^+$. This proves the proposition for $p_t^+$.\\ By repeating the same considerations for the integral $\int_{-\infty}^{-Q}dx\, e^{ax}\nu_t(x)$ we can show that also $p_t^-$ is non increasing in $t$ 
	\end{proof}

	\bigskip

	Thanks to the monotonicity of $p^+_t$ and $p^-_t$, we can easily identify the strip of regularity for any increment  $f_t - f_s$:
	its  characteristic function $\phi_{s,t} $ is analytic  when $ \Im (u) \in (-(p^+_t + 1), p^-_t)$ for any $s\in[0,t)$. \\
	\citet{Lewis2001} obtains the CDF, 
	shifting the integration path within the characteristic function horizontal analyticity strip. 
	The shift is $- i \; a$ with  $a$ a real constant s.t. $a \in (-p^-_t, p^+_t +1)$.
	Lewis deduces this formula using the properties of contour integrals in the complex plane.
	
	The CDF $P(x)$ of an additive process increment is \citep[see e.g.,][th.5.1]{Lee2004option}
	\begin{equation}
	P(x) = R_a - \frac{e^{-a x}}{\pi} \int_0^\infty du\, Re\left[\frac{e^{-iux}\phi_{s,t}(u-ia)}{i\,u+a}\right]\;\; , 
	\label{eq:Lewis CDF}
	\end{equation}
	where  \begin{equation*}
	R_a=\begin{cases}
	1\quad &\quad 0<a<p_t^++1\\
	\frac{1}{2} \quad &\quad a=0\\
	0\quad &\quad - p_t^-<a<0
	\end{cases} \;\;.
	\end{equation*}

	The case with no shift ($a=0$)	is the Hilbert transform: it has been considered in several studies in the financial literature on MC pricing 
	\citep[see e.g.,][]{chen2012simulating,BallottaKyriakou2014}. 
	In the Hilbert transform case, the singularity in zero in the integration should be taken into account as a  Cauchy principal value; 
	as already emphasized by \citet{BallottaKyriakou2014}, the method could be not robust  enough for applications in the financial industry:
	they have suggested a regularization technique that introduces an additional error source, while the Lewis method we consider here is exact (cf. also figure \ref{figure:error_digital} for a comparison between the CDF error with Lewis formula and the Hilbert trasform method).
	
	\bigskip
	
	In the following, we focus on $a > 0$: this is a default choice in the equity case  because $p_t^+\geq p_t^-$ is consistent with the negative equity skew \citep[see e.g.,][Section 7.4, p.26]{Lee2004option}. We derive an approximation formula and its error bounds (in sections \ref{subsection::error_CDF} and \ref{subsection::error_sources}). Similar results hold for $a<0$.
	
	We approximate the  Fourier transform with a discrete Fourier transform $ \hat{P}(x) $
\begin{equation}\label{eq:P_hat}
		\hat{P}(x) := 1 - \frac{e^{-a x} }{\pi}\sum_{l=0}^{{ N }-1} h\,Re \left[ \frac{e^{-i (l+1/2)h x}\phi_{s,t}((l+1/2)h-i\,a)}{i\,(l+1/2)h+a} \right]\;\;,
\end{equation}
	where $h$ is the step size in the Fourier domain and ${ N }$ is the number of points in the grid.  \\
	
	
	To implement the MC method,	
	we need the CDF function 
	for a large number of values in a regular grid with step size $\gamma$.
	An algorithm that is computationally efficient is the fast Fourier transform \citep[see][for a detailed analysis of the method in derivative pricing]{Lee2004option}:
	it  involves Toeplitz matrix-vector multiplication \citep[see e.g.,][ch.12]{NumericalRecipes1989} and relies on an additional requirement for $N$,
	whose simplest choice is $N = 2^M$ with $M\in \mathbb{N}$; 	hereinafter, we consider an $N$ within this set.	
	The main advantage of the method is that the computational complexity of the FFT is $O(N \log_2 N)$ when computing one time-increment.
	Moreover, with an FFT, it holds the relationship
	\[
	\gamma \, h = \frac{2 \pi}{N} \; \; ;
	\]
	i.e., for a given number $N$ of grid points, the step size in the Fourier domain $h$ fixes the step size 
	$\gamma$.\footnote{To avoid this constraint, one can consider the fractional fast Fourier transform \citep{Chourdakis2005} 
		instead of the standard FFT. We have verified that the additional computational cost of the former method is not justified in the CDF simulation described in this paper.}

	\subsection{CDF error sources}
	\label{subsection::error_CDF}
	
	The numerical  Fourier inversion is subject $i)$ to a discretization error, because the integrand is evaluated only at the grid points, and $ii)$ to a range error, because we approximate with a finite sum. 
	
	\bigskip
	
	{\bf Assumption.}
	$\forall\, t>s\geq 0$	there exists $B>0$, $b>0$ and $\omega>0$ such that, for sufficiently large $|u|$, the following bound for the absolute value of the characteristic function holds
	\begin{flalign*}
	&&	|\phi_{s,t}(u-i\,a)|<Be^{-b \, |u|^\omega} \;, \qquad \qquad\qquad\qquad \forall a\in (0,p_t^++1)
	\hspace*{\fill} 
	&&& \clubsuit
	\end{flalign*}
	
	\bigskip
	

	Leveraging on the Assumption, we can estimate the explicit bound for the bias in terms of the step size $h$ and the number of grid points ${ N }$, as shown in the next proposition.
	The result in the next proposition improves the known bounds for numerical errors when computing the CDF \eqref{eq:Lewis CDF}, via a discrete Fourier transform, and 
	indicates an optimal integration path that minimizes this error bound.
	
	\begin{proposition}
		\label{lemma:bound_cf} 
		If the Assumption holds, then 
		\begin{enumerate}
			\item	the numerical error $|P(x) - \hat{P}(x)|$  for the CDF  is  bounded by
			\begin{align}
			{\cal E}^{CDF}_{h, M} (x) =	
			\frac{Be^{- x \, (p^+_t+1)/2 }}{\omega  \pi  } 
			\Gamma \left[0,b \left( { N }\,h \right)^\omega \right]
			+\frac{e^{-\pi (p^+_t+1) /h}+e^{-\pi (p^+_t+1) /h- (p^+_t+1)\,x} \phi_{s,t}(-i \, (p^+_t+1))}{1-e^{-2\pi (p^+_t +1)/h}}\;\;,
			\label{eq:ErrorCDF}
			\end{align}
			where $\Gamma(z,u)$ is the upper incomplete gamma function and
			\[
			\Gamma \left[0,b \left( { N }\,h \right)^\omega \right]=O\left(({ N }\,h)^{-\omega}e^{-b \, ({ N }\,h)^{\omega}}\right) \;\; ;
			\]
			\item the (optimal) bound holds selecting the shift $a$ in (\ref{eq:Lewis CDF}) equal to $(p^+_t+1)/2$.
		\end{enumerate}
	\end{proposition}
	\begin{proof}
	We bound the range and the discretization error separately.

\medskip

First, we bound the CDF range error, i.e. the error we introduce considering the integral (\ref{eq:Lewis CDF}) in the range $(0,Nh)$.
Fix $h$, it exists $N\in \mathbb{N}$ s.t. 
\begin{align*}
	&\left|P(x)-\left(1-\frac{e^{-a x} }{\pi} \int_{0}^{{ N }\,h}du\; Re \left[ e^{-iux}\frac{\phi_{s,t}(u-ia)}{iu+a} \right]\right)\right| \nonumber
	\\&< \frac{B \,e^{-a x}}{\pi} \int_{{ N }\,h}^{\infty}du\; \frac{\, e^{-b\; u^{\omega}}}{u} 
	=
	\frac{B\,e^{-a x} }{\omega\pi } 
	\Gamma \left[0,b \left( { N }\,h \right)^\omega \right] =O\left(({ N }\,h)^{-\omega}e^{-b \, ({ N }\,h)^{\omega}}\right)\;\;.
\end{align*}

The first inequality is due to $|iu+a|>u$ for $a>0$
and to the fact that $|\phi_{s,t}(u-ia)|\leq Be^{-b\; u^{\omega}}$ for sufficiently large values of $u$, thanks to the Assumption. Notice that in the range error the order of the exponential decay does not depend on $a$. Below, we prove that the choice of $a$ determines the exponential decay of the discretization error: thus, its choice is crucial to get the optimal error bound.

\medskip

Second, we bound the CDF discretization error.\\
By theorem 6.2 of \citet{Lee2004option}, we have that for any $a, \,p$ s.t. $0<a<p<p_t^++1$
\begin{align*}
	&\left|\frac{e^{-a x} }{\pi} \int_{0}^\infty du\;e^{-iux}\frac{\phi_{s,t}(u-ia)}{iu+a}-\frac{e^{-a x} }{\pi}\sum_{l=0}^{{ N }-1} h\,Re \left[ \frac{e^{-i (l+1/2)h x}\phi_{s,t}((l+1/2)h-i\,a)}{i\,(l+1/2)h+a} \right] \right|\nonumber\\ &\leq  \frac{e^{-2\pi a /h}}{1-e^{-4\pi a /h}}+\frac{e^{-2\pi(p-a) /h-p\,x}}{1-e^{-4\pi(p-a) /h}} \phi_{s,t}(-ip)  \;\; ,
\end{align*}
where $\phi_{s,t}(-i \, p)$ is well defined because $0<p<p^+_t+1$. 

We select $a$ and $p$ to minimize the discretization error. Notice that, for a sufficiently small $h$, the leading terms in the bound on the discretization error are $e^{-2\pi a /h}$ and $e^{-2\pi(p- a) /h-p\,x}$. 
Hence, for a given $p$ the best choice of $ a$ is 
\[
\hat{a}= \frac{p}{2} \left( 1 -\frac{x}{\pi} \, h \right) \;\;.
\]
This last quantity, for a sufficiently small $h$, is close to $p/2$ for any finite $x$. 
Thus, to minimize the discretization error, we select $ a=p/2$. Then, $p$ can be chosen to its maximum value $p^+_t+1$ and the upper bound becomes
\begin{align*}
 \frac{e^{-\pi (p_t^++1) /h}+e^{-\pi (p_t^++1) /h-(p_t^++1)\,x}\phi_{s,t}(-i(p_t^++1))}{1-e^{-2\pi(p_t^++1) /h}}\;\;. 
\end{align*}
With the selection of $a=(p^+_t +1)/2$ and combining the bounds on the range and discretization errors, the thesis follows
	\end{proof}

	The first term of ${\cal E}^{CDF}_{h, M} (x)$ accounts for the range error in the numerical inversion, while the second one accounts for the discretization error.\footnote{ It is possible also to obtain an error bound even when the Assumption does not hold.
		Equation \eqref{eq:ErrorCDF} can be extended to the case where the characteristic function has an asymptotical polynomial decay $|\phi_{s,t} (u-i\,a)| \le B \, |u|^{-b}$, with $b>0$:
		in this case,  the range error decays only as a power of $u$ due to the polynomial decay of the characteristic function \citep[see e.g.,][eq.(14), p.1099]{BallottaKyriakou2014}.
		However, in practice, when pricing exotic derivatives, the exponential decay of the characteristic function is a good reason for model selection.
	} 
	It is possible to prove, following the same steps of \textbf{proposition  \ref{lemma:bound_cf}}, that in the $a<0$ case the leading term in ${\cal E}^{CDF}_{M} (x)$ is $\exp(-\pi p^-_t/h)$.\footnote{In this case, the optimal shift is $a=-p_t^-/2$.} From this result, we can observe that it is convenient to use $a>0$ if $p_t^++1\geq p_t^-$ and $a<0$ otherwise. 
	In the financial literature, error estimations have been proposed when approximating a CDF via a discrete Fourier Transform
	\citep[see e.g.,][]{Lee2004option,chen2012simulating, BallottaKyriakou2014}. 
	The bound in {\bf proposition \ref{lemma:bound_cf}} extends these results to the Lewis-FFT case, showing how to select the optimal integration path in the Lewis formula
	\eqref{eq:Lewis CDF}
	to minimize the exponential decay of the error.
	Our approach eliminates the source of error originating from the pole in the origin \citep[see e.g.,][eq.(4), p.1097]{BallottaKyriakou2014}, improving the CDF error. 
	Moreover, selecting the optimal path, CDF error is even better than the one proposed by \citet[][th.2.1, p.14:6]{chen2012simulating} deduced via the sinc expansion technique. 
	The leading term in the discretization error in theorem 2.1 of \citet{chen2012simulating}  goes as $\max(e^{-\pi \,p_t^-/h},e^{-\pi \,(p_t^++1)/h})$, while, in our case, the error goes as the minimum of the two terms. Hence, we improve the discretization error of \citet{chen2012simulating} in all cases.\footnote{\citet{baschetti2022sinc} point out that the symmetry in the real and imaginary components of the Hilbert transform allows to compute the CDF only N/2 times when the FFT grid size is N. This observation becomes relevant for situations where computing the characteristic function is computationally demanding. However, in the case of additive processes, characteristic functions are analytic and very fast to compute.}

	\bigskip
	
	We desire to get a small approximation error increasing $N$ and decreasing $h$.
	However, let us observe that, if one takes the limit $h \to 0$ and $N \to \infty$  keeping $N h$ fixed, then the range error bound does not decrease.
	Thus, our interest is to select $h = h(N)$ so that the discretization and the range errors have about the same order.
	Expression \eqref{eq:ErrorCDF} allows us to determine the size $h$ and the number $N$ such that the two sources of CDF error are comparable: 
	we can impose that
	$\exp(-\pi (p^+_t+1)/h) = \exp(-b\, (N \, h)^\omega) $, i.e. we select
	\[
	h(N) = \left( \frac{\pi \, (p^+_t+1)}{b} \dfrac{1}{N^\omega} \right)^{\displaystyle {1}/({\omega +1})} \;\; .
	\] 
	We define 
	\begin{equation}\label{equation::error_cdf_M}
	{\cal E}^{CDF}_{M} (x) := {\cal E}^{CDF}_{h(2^M),M} (x)
	\end{equation}
	the error in this case. ${\cal E}^{CDF}_{M} (x)$ in (\ref{equation::error_cdf_M}) is the relevant estimation of the CDF error that we use in practice:
	with this selection of $h$, the total CDF error is 
	$O(N^{-\omega/(1+\omega)} ) \exp(-b N^{\omega/(1+\omega)}) $ and decays almost exponentially 
	as we increase $N$; moreover, the step size $\gamma =2\pi/(h{ N }) = O(N^{-1/(1+\omega)})$.
	\section{The simulation method}
	\label{sec:SimulationMethod}
	Knowing the CDF approximation $\hat{P}$ in \eqref{eq:P_hat}, we can sample from this
	distribution by  inverting  $\hat{P}$, i.e. by
	setting $X =\hat{P}^{-1} (U)$, with $U$ an uniform r.v. in $[0,1]$.
	
	From the Fourier inversion, we obtain an estimate of $\hat{P}$ on a grid of $N$ points with step $\gamma$.
	As pointed out by \citet[][sec.3, pp.1614-1615]{glasserman2010sensitivity}, an adequate inversion requires to impose that $\hat{P}$ is 
	i) increasing and 
	ii) inside the interval [0,1].
	Thus, 
	it is convenient to work with a subset of the grid of $N$ points. 
	We truncate the CDF between $x_0<0$ and $x_K>0$, such that the two conditions hold, 
	and we consider the equally spaced grid (with step $\gamma$) $x_0<x_1<...<x_K$ with $K < { N }$. \\
	Simulating a r.v. via a numerical CDF (i.e. coupling the Fourier transform with a MC simulation), requires a numerical inversion that is realized with an interpolation method.  
	As already discussed in section \ref{sec:Overview}, differently from  the existing financial literature \citep[see e.g.,][]{glasserman2010sensitivity,chen2012simulating, FengLin2013}, the 
	proposed method is based on spline interpolation. 
	In the next subsection, we discuss the key idea behind this choice of the interpolation method.

	\subsection{Simulation error sources: truncation and interpolation}
	\label{subsection::error_sources}
	

	Besides numerical inversion error of the CDF, two are the error sources in the MC, when pricing a contingent claim: truncation and interpolation of the CDF.	\\
	Let us consider the expected value $\mathbb{E}V(f_t-f_s)$, with $V(x)$ a derivative contract with a pay-off  differentiable everywhere except in $n_V$ points. 
	It can be proven, similarly to \citet[][ th.4.3, p.14:11]{chen2012simulating}, that
	the pricing error\footnote{The upper bound on the bias ${\cal E}$ can be trivially extended to a payoff with a finite number $n$ of monitoring times. The most relevant case, for $n=1$, will be discussed in detail in subsection \ref{subsection::European_options}.} using the Lewis-FFT method with linear interpolation 
	is
	\begin{align}
	\label{eq:Bias_error}	
	{\cal E} :=& \int^{\infty}_{-\infty} dx \, V(x) \, \left[p(x) - \hat{p}(x) \right]\\
	<
	& \left(|V(x_0)|+|V(x_K)|+(2K+n_V) \sup_{x\in (x_0,x_K)}|V(x)|+2\sup_{x\in (x_0,x_K)}|V'(x)|  \right) {\cal E}^{CDF}_{M}(x_0) \label{eq:Error_FFT} \\
	& + \frac{\phi_{s,t}^-}{2\pi}\left(\frac{|V(x_K)|e^{x_K p^-_t}}{|p^-_t|}+ \int_{x_K}^{\infty} dx\,V(x) e^{x\, p^-_t} \right) +
	\frac{\phi_{s,t}^+}{2\pi}\left(\frac{V(x_0)e^{x_0 (p^+_t+1)}}{p^+_t +1}+ \int_{-\infty}^{x_0} dx\,V(x) e^{x \, (p^+_t + 1)} \right)\label{eq:truncation_CDF}\\
	&+ \frac{\gamma^2}{2\pi} (x_K-x_0) 	\sup_{x\in (x_0,x_K)}|V'(x)| \int_{\mathbb{R}} | du\,u\, \phi_{s,t}(u)|\label{eq:interp_CDF} \;\;,
	\end{align}
	
	where $p(x)$ is the probability density function of $f_t-f_s$, $\hat{p}$ its estimation and
	
	\begin{equation*}
	\phi_{s,t}^- := \lim_{a\to p_t^-}\int_{\mathbb{R}}du\, |\phi_{s,t}(u-ia)| \quad \&\quad \phi_{s,t}^+:=\lim_{a\to p_t^++1}\int_{\mathbb{R}}du\, |\phi_{s,t}(u-ia)|\;\;.
	\end{equation*}
	
	Three are the components of the bias error \eqref{eq:Bias_error} when pricing a derivative: an error related to the numerical approximation of the CDF (\ref{eq:Error_FFT}),
	a truncation error \eqref{eq:truncation_CDF} and an interpolation error  \eqref{eq:interp_CDF}. Let us consider each error source separately.\\
	
	First, the error related to the numerical approximation of the CDF in  \eqref{eq:Error_FFT} is proportional to ${\cal E}^{CDF}_{M}(x_0)$: we have discussed in the previous section how to select the integration path and $h$ in order to minimize it.
	\smallskip\\
	Second, we can always choose $x_0$ and $x_K$ s.t. the truncation error is negligible for all practical purposes. We select these points s.t. $\hat{P}(x_0)<10^{-10},\,1-\hat{P}(x_K)>10^{-10}$ \citep[as suggested by][eq.5]{baschetti2022sinc}. We notice that the range $(x_0,x_K)$ scales with $\sqrt{t-s}$. In figure \ref{figure:PDF_norm}, as an example, we plot the one-day and one-year normalized probability density functions of the additive process used in the numerical experiments of section \ref{sec:NumericalResults}. As expected, the one-day density is significantly more concentrated around zero than the one-year density when considering a constant $x$ (on the right). Conversely, the ranges of the two densities look similar when considering the rescaled $x/\sqrt{t-s}$ on the abscissa.\footnote{In extensive numerical experiments, we have observed that when choosing $x_0=-x_K$ and $x_K$ the nearest point to $5\sqrt{t-s}$ the above condition on $\hat{P}(x_0)$ and $\hat{P}(x_k)$  is always satisfied.} Moreover, to further improve the method accuracy (in particular when $M$ is small), 
	we introduce
	 an exponential extrapolation for the CDF tail below $x_0$ and above $x_K$.

	\begin{center}
		\begin{minipage}[t]{1\textwidth}
			\centering
			{\includegraphics[width=1\textwidth]{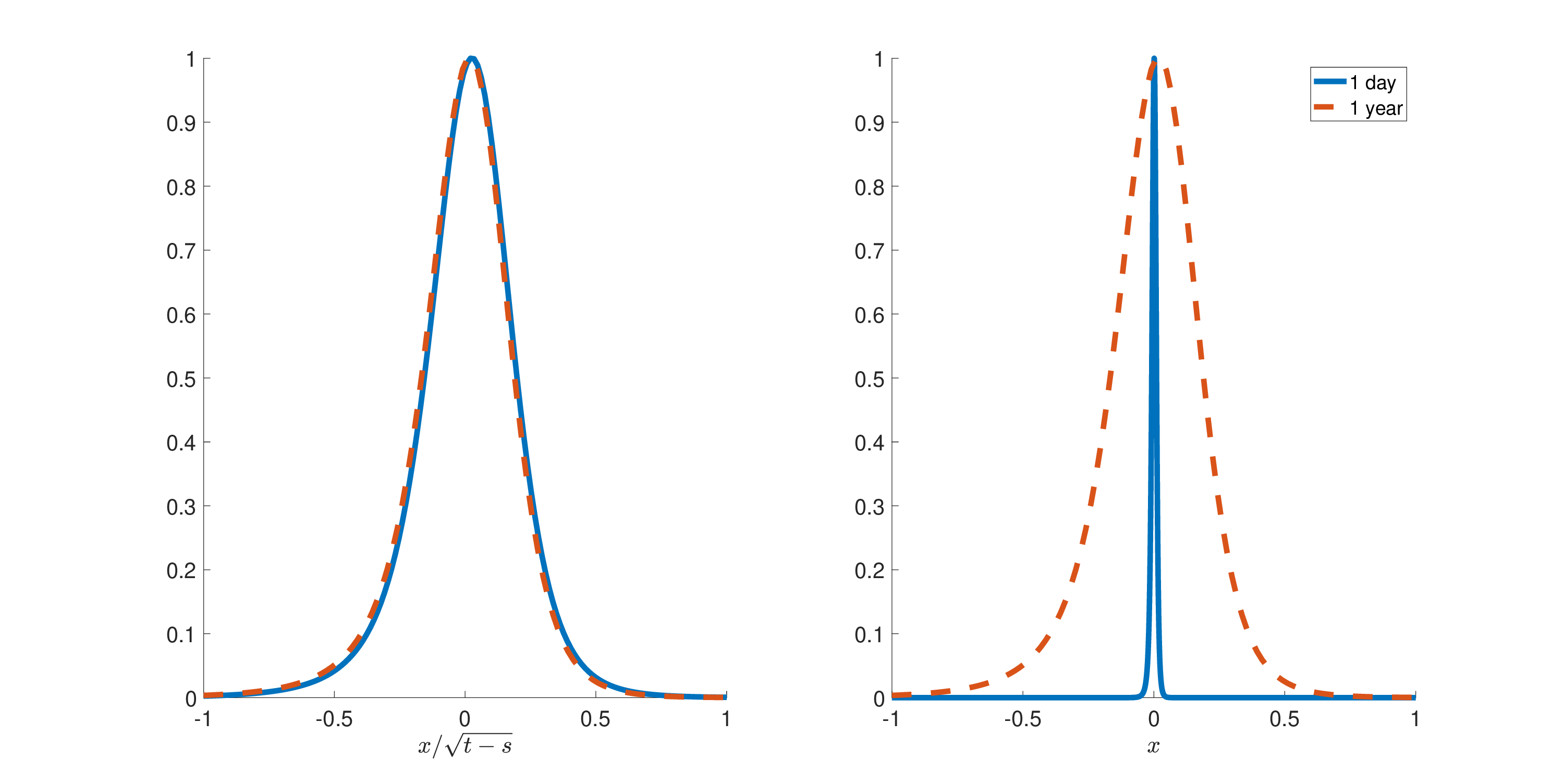} }
			\captionof{figure}{\small One-day and one-year normalized probability density functions of the additive process that we use in the numerical experiments of section \ref{sec:NumericalResults} with $s=0$. On the right, we see that, as expected, the one-day density is significantly more concentrated around zero than the one-year density. Conversely, on the left, we see that the ranges of the two densities  wrt to the rescaled $x/\sqrt{t-s}$ are similar. Notice that both probability density functions have been divided by their respective maximum for visualization purposes.
				\label{figure:PDF_norm}} 
		\end{minipage}
	\end{center}

	\smallskip	
	
	Finally, the bias  associated with the linear interpolation, when computing the option value, is quadratic in the grid spacing $\gamma$;
	this turns out to be the most significant source of error in most cases, as shown in the next section.
	It is well known that linear interpolation error goes as $\gamma^2$ \citep[see e.g.,][eq.(8.26), p.339]{quarteroni2010numerical}.
	For this reason,  in this paper, we propose a spline interpolation method. In this latter case, it is known that the bias goes, at least, as $\gamma^4$ as shown in \citet{hall1976optimal}.\\


	As already emphasized by \citet[][Sec.3, p.1615]{glasserman2010sensitivity}, to sample $X$ from $\hat{P}(x)$ with a linear interpolation, 
	after having generated $U$, a r.v. uniformly distributed in $[0,1]$, one should 
	\begin{enumerate}
		
		\item  select the index $j$ for which $ \hat{P}(x_{ j-1}) \le  U < \hat{P}(x_{ j})$;
		\item for each $j$ determine the linear interpolation coefficients  $c_{0,j}^L$ and $c_{1,j}^L$ \begin{equation*}
		c_{0,j}^L:=\frac{ x_{ j} \, \hat{P}(x_{ j}) - x_{ j-1} \, \hat{P}(x_{ j-1})}{\hat{P}(x_{ j}) - \hat{P}(x_{ j-1})}\quad\text{and}\quad c_{1,j}^L:=\frac{\gamma} {\hat{P}(x_{ j}) - \hat{P}(x_{ j-1})}\;\;;
		\end{equation*} 
		\item compute 
		\[
		X = c_{0,j}^L+c_{1,j}^L\, U\,\,\;.
		\]
	\end{enumerate}
	Let us  discuss the computational cost of each step  when sampling ${\cal N}_{sim}$ observations. The first step relies on a nearest neighborhood algorithm with an average computational cost proportional to ${\cal N}_{sim}\times \log_2{{\cal N}_{sim}}$ \citep[see e.g.,][p.11]{cormen2001introduction}\footnote{The computational cost estimation is for the \textit{merge sort} algorithm. Since \textit{merge sort} is a recursive algorithm it could be necessary, for memory efficiency, to recur to \textit{an insertion sort} algorithm which computational cost is roughly proportional to ${\cal N}_{sim}^2$ \citep[see e.g.,][p.11]{cormen2001introduction}.}. The second step  cost is proportional to $6 K$. Finally, the last step is proportional to ${\cal N}_{sim}$.
	
	Whereas step $1$ is shared by both interpolation methods, steps $2$ and $3$ differ between spline and linear interpolations.
	In step 2, the additional computational cost of considering spline interpolation boils down to the cost of solving a $K+1$-dimensional linear system with a tridiagonal matrix to determine the spline coefficients $\{c_{q,j}^S\}_{q=0}^3$, cf. \citet[][ch.8]{quarteroni2010numerical}, i.e. the cost is 
	$8 K-7$ \citep[ch.7, p.391]{quarteroni2010numerical}. As for step $3$, the cost of computing the spline interpolation of $U$ is still proportional to ${\cal N}_{sim}$.
	It is clear that for a sufficiently large number of simulations ${\cal N}_{sim}$ and for ${\cal N}_{sim}>>K$, for both methods, the most relevant contribution in the computational cost is the one due to step $1$, the nearest neighborhood algorithm.\\

	We perform numerical experiments to compare linear and spline interpolation. We observe that, if the number of simulations is significantly above the grid dimension $K$, the spline interpolation is as expensive as the linear interpolation. 
	Moreover, in table \ref{table::Computational spline}, we compare the computational cost of linear interpolation and spline interpolation. We consider a grid of size $K=10^4$ and  ${\cal N}_{sim}=10^5$. In this case, the spline cost is just 10\% more than the linear one.   The case considered in table \ref{table::Computational spline} is a particularly unfavorable situation, when comparing spline interpolation with linear interpolation: a large grid size $K=10^4$ and a small number of simulations ${\cal N}_{sim}=10^5$. In this case steps 1, 2 and 3 computational times are comparable while, in practice, most of the computational costs are absorbed by the nearest neighborhood algorithm.
	For reasonable values of $M$ (e.g. for $M\leq 15$), the dimension of the grid $K$ is always well below $10^4$. Thus, for all values of $K$ and ${\cal N}_{sim}$ (${\cal N}_{sim}\geq 10^6$) used in practice the incremental cost between Lewis-FFT (with linear interpolation) and Lewis-FFT-S (with spline interpolation) is negligible. 
	\begin{table}
		\centering
		\begin{tabular}{|l|ccc|}
			\hline
			Algorithm& Nearest neighborhood& Linear interpolation&Spline interpolation \\
			time [ms] &1.08 &1.13 &1.27\\
			
			\hline
		\end{tabular}

		\caption{\small Computational cost in milliseconds [ms] for the nearest neighborhood (nn), 
			the linear interpolation including nn, and the spline interpolation including nn. We consider a grid size $K=10^4$ and  ${\cal N}_{sim}=10^5$ simulations. Even considering a low number of simulations and a grid size $K$ one order of magnitude above what is used in practice (in the Lewis-FFT-S case $K$ is of order $10^3$)  the spline simulation cost is just 10\% more than the linear simulation one. \label{table::Computational spline}}
	\end{table}

	\subsection{A simulation benchmark: the Gaussian approximation}
	
	In this subsection, we show how to generalize the GA method for additives in an efficient way,
	when a monotonicity property holds for the L\'evy measure 
	and then the ziggurat method \citep{marsaglia2000ziggurat} can be applied.
	
	A generic additive process may have an infinite number of jumps, most of them being small, over an arbitrary finite time horizon,
	making the simulation of such a process often nontrivial. Defining  $\nu_t$ the additive process jump measure  \citep[see e.g.,][def.8.2, p.38]{Sato1999}, the jump measure of the additive process increment $f_t-f_s$ is $\nu_t-\nu_s$.\\

	\citet{EberleinMadan2009}, in their study on simulation of additive processes,
	consider only a class of additive processes (Sato processes):
	their approach consists in discarding the small jumps that in absolute value are below a given threshold $\epsilon$.
	It is well known, in the L\'evy case, that such an approach is accurate only if there are not too many
	small jumps \citep[see e.g.,][]{Cont}. 
	Alternatively, 
	the small
	jump component of an additive process may be approximated by a Brownian motion \citep{Asmussen2001approximations}.\\
	
	Once the jump measure of the increment (between time $s$ and time $t>s$) is truncated, we have  
	$i$) to draw a Poisson
	number of positive and negative jumps  and
	$ii$) to simulate separately positive jumps from the probability density  $m^+_{s,t}$ and negative jumps from  the probability density $m^-_{s,t}$, where \begin{equation}
	m^+_{s,t}(x) := \mathbb{I}_{x>\epsilon}\frac{\nu_t(x)-\nu_s(x)}{\int_\epsilon^\infty dz (\nu_t(z)-\nu_s(z))} \quad \& \quad  m^-_{s,t}(x) :=\mathbb{I}_{x<-\epsilon} \frac{\nu_t(x)-\nu_s(x)}{\int_{-\infty}^{-\epsilon} dz (\nu_t(z)-\nu_s(z))} \;\;. \label{eq:m}
	\end{equation}
	
	To sample positive and negative jumps is extremely costly because often it is not possible to compute explicitly the integrals of $m^+_{s,t}$ and $m^-_{s,t}$. \\
	
	When $m^+_{s,t}(x)$ is  non increasing  in $x$ and $m^-_{s,t}(x)$ is non decreasing in $x$ $\forall s,t$ s.t. $0\leq s < t$, 
	a faster methodology -for sampling from a known distribution without inverting numerically its integral- is available: the ziggurat method of \citet{marsaglia2000ziggurat}. 
	This method is applicable to probability density functions that are bounded and monotonic. We can apply the algorithm separately to negative and positive jumps. Notice that the density functions are bounded because we have truncated the small jumps. The ziggurat method covers a probability density with $N_{ret}$  rectangles with equal area and a base strip. The base strip contains the tail of the probability density, it is built s.t. it has the same area of the rectangles. The method is composed of two building blocks: first, the rectangles with equal area are identified; second, the random variable is simulated either from a rectangle or from the base strip. Only in the latter case, an inversion of the integral is needed. $N_{ret}$ is a key parameter because it controls the trade-off, in terms of computational time, between the inversion and the construction of the rectangles.\\

	With respect to \citet{EberleinMadan2009}, to reduce the bias of the method, 
	we also consider the Gaussian approximation of \citet{Asmussen2001approximations}.

	\section{Numerical results}
	\label{sec:NumericalResults}
	
	Financial applications provide an important motivation for this study.
	We show that the proposed Monte Carlo technique for additive processes can price path-dependent options fast and accurately. The computational time is comparable to the case with simple Brownian motion dynamics.
	
	We are interested in simulating a discrete sample path of the  process over a finite time horizon:
	we are only concerned about the values of an additive process on such a discrete-time grid.
	This arises from situations where only discrete values of the process are concerned as in \citet{chen2012simulating,BallottaKyriakou2014}
	(e.g., they consider discrete barrier, lookback, and Asian options).

	\smallskip
	
	The case of an additive normal tempered stable (ATS) is discussed in detail.
	ATS processes present several advantages: they calibrate accurately equity implied volatility surfaces and, in particular, they capture volatility skews \citep[see e.g.,][]{azzone2019additive}. We model the forward price at time $t$ with maturity $T$  as an exponential additive $F_t(T)=F_0(T)e^{f_t}$, where $f_t$ is the ATS process and $F_0(T)$ is the forward price at time 0. The ATS characteristic function and L\'evy measure are reported in appendix A (cf. equations (\ref{laplace}-\ref{eq:nu_t})).
	
		\smallskip
	
	The Lewis-FFT-S method can be used for the ATS because, in the next proposition, 
	we prove that the Assumption in section \ref{subsection::error_CDF} holds for this class of additive processes. Moreover, we prove that the Assumption holds for the two other classes of additive processes considered in the literature for option pricing: additive logistic processes \citep[][]{carr2021additive} and Sato processes \citep{carr2007self}. 
	\begin{proposition}
		The Assumption (cf. section \ref{subsection::error_CDF}) holds for\begin{enumerate}
			\item ATS processes with $\alpha \in (0,1)$;
			\item additive logistic processes \citep[][]{carr2021additive};
			\item Sato processes with characteristic function  $\phi_{t}(u)$, for $t=1$, that decays exponentially \citep{carr2007self}.\end{enumerate} \label{proposition::assumptions_self}
	\end{proposition}
	\begin{proof}
	We prove the thesis for the ATS.\\
We observe that, by the condition $(a)$ on $g_1(t)$ and $g_2(t)$ of theorem \ref{theorem:f_Additive}, we have 	
\[
g(t):= - (g_1(t)+ g_2(t)) = \sqrt{{\left(1/2+\eta_t\right)^2+2(1-\alpha)/( k_t\,{\sigma}^2_t)}}
\]  
is non increasing wrt $t$. Hence, thanks to the condition $(a)$ on $ g_3(t)$ of theorem \ref{theorem:f_Additive}
\be
\frac{t}{k_t^{1-\alpha}}\sigma_t^{2\alpha} \; \; \text{is increasing in } t \;\;.
\label{eq:coda}
\en

We have to show that, given $s$ and $t$, there exists $B>0$, $b>0$, and $\omega>0$ such that, for sufficiently large $|u|$, the Assumption holds for the characteristic function of ATS.

\smallskip

We choose $\log(B)>\frac{1-\alpha}{\alpha}\left( \frac{t}{k_t} -\frac{s}{k_s}\right)$, 
$0<b<\frac{(1-\alpha)^{1-\alpha}}{2^\alpha \alpha}\left(\frac{t}{k_t^{1-\alpha}}\sigma_t^{2\alpha}-\frac{s}{k_s^{1-\alpha}}\sigma_s^{2\alpha}\right)$ and $0<\omega<2\alpha$.\\ Notice that it is possible to fix $b>0$, because \eqref{eq:coda} holds. Moreover, the imaginary part of the exponent in (\ref{laplace}) does not contribute to $B$, because the absolute value of the exponential of an imaginary quantity is unitary.\\
For sufficiently large $|u|$, and for $s<t$ $|\phi_{t,s}(u-i\,a)|$ goes to zero faster than $Be^{-b\; |u|^{\omega}}$ because  $\log\phi_{t,s}(u-i\,a)$ is asymptotic to 
\[
-\frac{(1-\alpha)^{1-\alpha}}{2^\alpha \alpha}\left(\frac{t}{k_t^{1-\alpha}}\sigma_t^{2\alpha}-\frac{s}{k_s^{1-\alpha}}\sigma_s^{2\alpha}\right) \, u^{2 \, \alpha}\;\;,
\]
that is negative due to \eqref{eq:coda} for $\alpha \in (0,1)$.\\

We prove the thesis for additive logistic processes. \\
\citet{carr2021additive} consider two additive logistic processes: the CPDA and the SLA.	The characteristic function of an additive logistic process at time $t$ is 
\begin{equation*}
	\phi_t(u)=\frac{{{\cal B}}(1+i\,\sigma_tu,c_t-i\,\sigma_t u)}{{B}(1,c_t)}\;\,,
\end{equation*}
where ${\cal B}$ is the beta function and $\sigma_t>0$ is non decreasing. For  the CPDA model $c_t=1-\sigma_t$ and $\sigma_t<1$ and for the SLA model $c_t=1$ \citep[cf.][prop.4.2]{carr2021additive}. \\
For sufficiently large $|u|$,
\begin{align*}
	&{{\cal B}}(1+i\,\sigma_tu,c_t-i\,\sigma_t u)=\frac{\Gamma(1+i\,\sigma_tu)\Gamma(c_t-i\,\sigma_t u)}{\Gamma(1+c_t)}
	\\\approx& \frac{\sqrt{2\pi}}{\Gamma(1+c_t)}(1+i\,\sigma_tu)^{1+i\,\sigma_tu-1/2}(c_t-i\,\sigma_tu)^{c_t-i\,\sigma_tu-1/2}\\=&\frac{\sqrt{2\pi}e^{i\,z_t(u)-1-c_t}}{\Gamma(1+c_t)} e^{\log\left(\sqrt{1+\sigma_t^2u^2}\right)(1-1/2)-\arctan(\sigma_tu/1)\sigma_tu+\log\left(\sqrt{(c_t)^2+\sigma_t^2u^2}\right)(c_t-1/2)-\arctan(\sigma_tu/c_t)\sigma_tu}\;\;,
\end{align*}
where the asymptotic approximation follows from Stirling's formula for the Gamma function $\Gamma(\zeta)$ when $\zeta\to \infty$ and $\arg \zeta<\pi$ \citep[see e.g.][p.257]{abramowitz1948handbook}, and $z_t(u)$ is a deterministic function. From this approximation, if $t>s$, $\log \left[\phi_t(u-i\,a)/\phi_s(u-i\,a)\right]$ is asymptotic to \begin{align*}
	&-|u|\left[\sigma_t\left(\arctan\left(\frac{\sigma_t|u|}{1+\sigma_ta}\right)+\arctan\left(\frac{\sigma_t|u|}{c_t-\sigma_ta}\right)\right)-\sigma_s\left(\arctan\left(\frac{\sigma_s|u|}{1+\sigma_sa}\right)+\arctan\left(\frac{\sigma_s|u|}{c_s-\sigma_sa}\right)\right)\right]\\&\leq-|u|\left[\left(\sigma_t-\sigma_s\right)\left(\arctan\left(\frac{\sigma_t|u|}{1+\sigma_ta}\right)+\arctan\left(\frac{\sigma_t|u|}{c_t-\sigma_ta}\right)\right)\right]\leq -|u|\left(\sigma_t-\sigma_s\right)\frac{3\pi}{4} =:-|u|\hat{b}\;\;,
\end{align*}
where the first inequality holds because $\frac{\sigma_t|u|}{1+\sigma_ta}$ and $\frac{\sigma_t|u|}{c_t-\sigma_ta}$ are non decreasing in $t$ and positive and the second holds for sufficiently large $|u|$.
Moreover, $\hat{b}>0$ because $\sigma_t$  is non decreasing in $t$.
Hence, we can set $B>0$ and $0<b<\hat{b}$ s.t. \begin{equation}
	|\phi_t(u)/\phi_s(u)|\leq Be^{-|u|b}\;\;,
\end{equation}
for sufficiently large $|u|$.\\

Finally, we prove the thesis for Sato processes.\\
If $\phi_1(u)$ decays exponentially as $e^{-\hat{b}|u|^w}$, with $\hat{b}>0$, then $
\phi_t(u) =\phi_1(ut^\zeta)$ decays as  $e^{-\hat{b}|u|^w \,t^{\zeta\,w}}$. It is possible to select $B>0$ and $0<b<\hat{b}(t^{\zeta\,w}-s^{\zeta\,w})$ s.t. $|\phi_{t,s}(u-i\,a)|=|\phi_t(u-i\,a) /\phi_s(u-i\,a) |<Be^{-b|u|^w}$  for $t>s$ 
	\end{proof}
	A brief comment on Sato processes can be useful. Thanks to the self-similarity of the processes, if a condition on the characteristic function holds for $t=1$ then it is satisfied also for all other time intervals.\footnote{\citet{EberleinMadan2009} consider also some characteristic functions with polynomial decay; in this case, the considerations in note 3 hold.}

	\smallskip
	
	In particular, for the numerical example, we focus on the power-law scaling ATS \citep[see e.g.,][p.503]{azzone2019additive} that is characterized by the parameters 	
	\[
	k_t=\bar{k} \; t^\beta, \qquad \eta_t=\bar{\eta} \; t^\delta,  \qquad\sigma_t=\bar{ \sigma}\;,
	\]
	where  $\bar{\sigma}, \bar{k}, \bar{\eta} \in \mathbb{R}^+$, and $\beta, \delta \in \mathbb{R}$. 
	This model description has been shown to be particularly accurate for equity derivatives. 
	Let us emphasize that, in the ATS case,  $p_t^+\geq p_t^-$, as shown in the next proposition, and then it is convenient to use $a>0$ (cf. section \ref{sec:Overview}).
	\begin{proposition}\label{proposition::assumptions}
		For ATS processes with $\alpha \in (0,1)$
		we have that $p_t^+\geq p_t^- $.
		
	\end{proposition}
	\begin{proof}
To identify $p_t^+$ and $p_t^-$, we apply the Lukacs theorem \citep[cf.][th.3.1, p.12]{Lukacs1972}.
At time $t$, the ATS characteristic function in equation (\ref{laplace})  is analytic on the imaginary axis $u=-i\,a$,  $a\in \mathbb{R}$ iff \[
1+\frac{k_t}{1-\alpha}\left(a\left(\frac{1}{2}+\eta_t \right)\sigma_t^2-\frac{a^2 \sigma_t^2}{2}\right)>0\;\;.
\]
By solving the second order inequality, we get \[
g_1(t)<a<-g_2(t)\;\;,
\]
with $	g_1(t)$ and $	g_2(t)$ defined in (\ref{eq:g_2}).
Hence, $p^+_t:=-g_2(t)-1$ and $p_t^-:=-g_1(t)$.\\ It holds that  $p^+_t\geq p_t^-$ because 
\[p^+_t-p_t^-=2\eta_t \geq 0\;\;\qedhere\eqno 
\]
	\end{proof}
	For all numerical experiments, we use the parameters reported in table \ref{table::parameters}: these parameters are consistent with the ones observed in market data. Moreover, 	for simplicity,	we consider the case with unitary underlying initial value and without 
	interest rates nor dividends:
	these deterministic quantities can be easily added to simulated prices without any computational effort.\footnote{We remind that, in this setting, the forward price $F_0(T)$ is equal to the spot price $S_0=1$.} 
	\begin{table}
		\centering
		\begin{tabular}{|ccccc|}
			\hline
			$\beta$&$\delta$&$\overline{k}$&$\bar{\eta}$&$\bar{ \sigma}$ \\
			\hline
			1&-1/2&1&1&0.2\\
			\hline
		\end{tabular}
		\caption{\small ATS parameters used in all numerical simulations. These selected parameters are consistent with the ones observed in market data. 
			\label{table::parameters}}
	\end{table}

	To evaluate the Lewis-FFT-S performances, we consider plain vanilla and exotic derivatives at different moneyness $x=\log(S_0/\kappa)$, where $\kappa$ is the strike price, and at different times to maturity. 
	In the rest of the section, to ensure that we verify the performance of the method on options in a relevant range of moneyness $x$, we consider $x$ in the range $\sqrt{t}(-0.2,0.2)$, where $t$ is the option time to maturity; deep out-of-the-money and in-the-money options are less informative on the method performances, as the option value is close to the intrinsic value.\\
	
	In subsection \ref{subsection::European_options}, we show how the Lewis-FFT-S (with spline interpolation) method 
	significantly outperforms the method with linear interpolation for European options, where - thanks to the closed formula - we can easily verify the accuracy of the numerical method. 
	In subsection \ref{subsection::Computational time}, we provide evidence that Lewis-FFT-S is extremely fast and it is less computationally expensive, by at least 1.5 orders of magnitude than the GA method. In subsection \ref{subsection::Discretly monitoring}, we price discretely-monitored Asian options, lookback options, and Down-and-In options with a time to maturity of five years. We also show that the Lewis-FFT-S is particularly efficient. The computational time needed to price path-dependent options with this method is just three times the computational time needed when using standard MC techniques for a geometric Brownian motion.
	
	\subsection{European options: accuracy} \label{subsection::European_options}
	
	In the following, the Lewis-FFT-S performances are assessed for the ATS process. First, we compare the accuracy of Lewis method and Hilbert transform to compute the CDF. Second, we show that, when using linear interpolation  the leading term in $(\ref{eq:Bias_error})$ goes as $\gamma^2$. Then, we improve the bound by considering spline interpolation (Lewis-FFT-S) and we discuss the excellent performances of the method for the ATS case. Thanks to FFT the  Lewis-FFT-S is particularly fast: computational time has the same order of magnitude of standard algorithms that simulate Brownian motions. Thanks to the spline interpolation, Lewis-FFT-S is also particularly accurate, for $10^7$ simulations and for any $M>9$, the maximum observed error is 0.03 basis points (bp).\\

	In figure \ref{figure:error_digital}, we compare the accuracy of the Lewis formula and the Hilbert transform method for inverting the CDF in terms of the mean absolute error (MAE) varying $M$ s.t. $N=2^M$. We consider the ATS case for the one month maturity and we invert the CDF on an interval $x_0,\,x_K$. To investigate the potential instability of the Hilbert transform due to the pole in the origin, we consider  both a small shift of $0.01\cdot h$ in the FFT grid in the Fourier space and the case of a perfectly symmetric grid.
	The Lewis method is more accurate than the Hilbert transform method both in the case of a shift in the Fourier space (on the left) and in the  symmetric case (on the right).  The plotted results clearly indicate that the Hilbert method is highly unstable and even a slight  shift in the Fourier space  can result in a significant increase of the error, up to six orders of magnitude.
		\begin{center}
		\begin{minipage}[t]{1\textwidth}
			\centering
			{\includegraphics[width=1\textwidth]{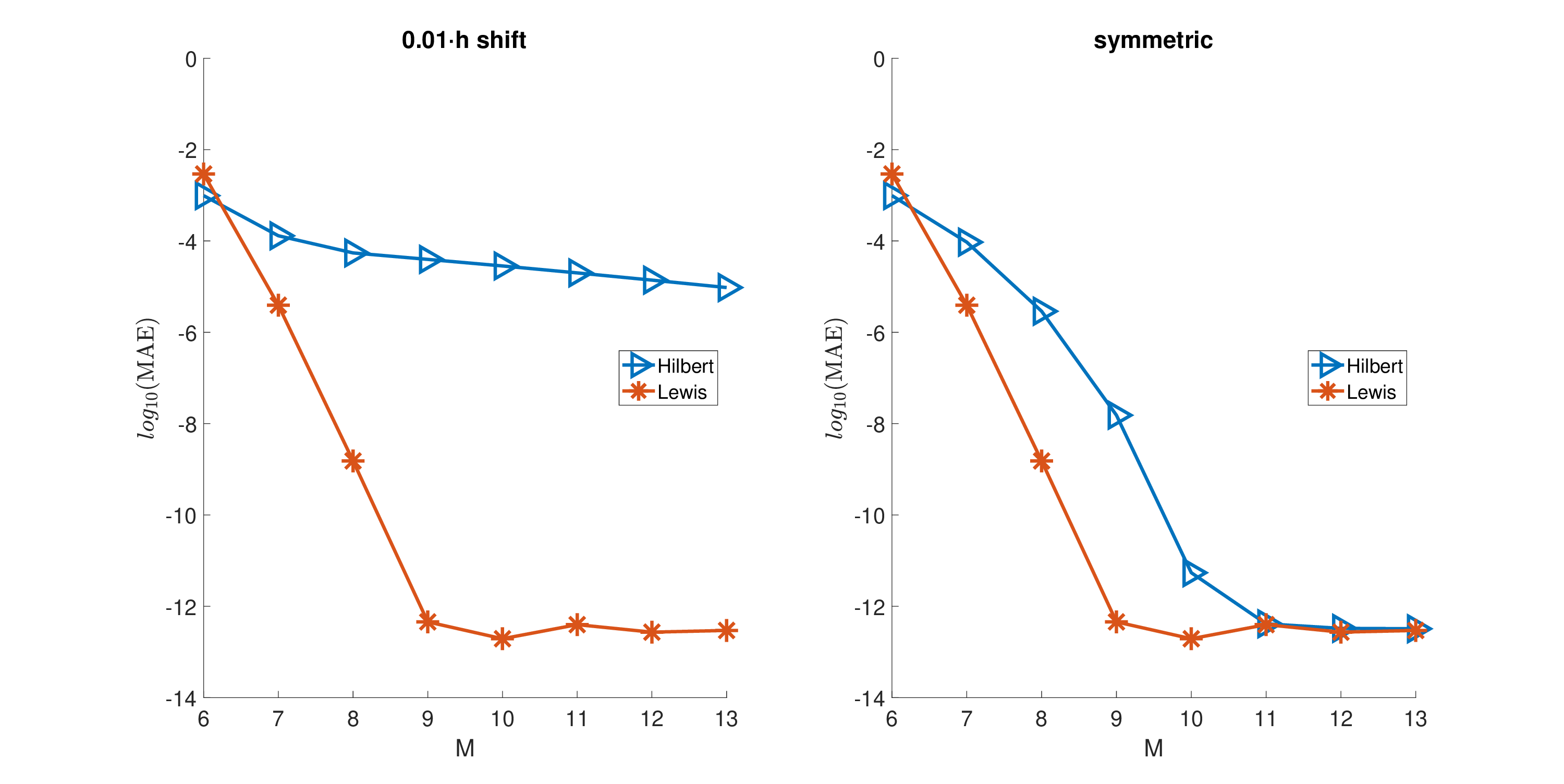} }
			\captionof{figure}{\small One-month mean absolute error (MAE) in the CDF (in log-scale) for the Lewis formula and Hilbert transform with a small shift of $0.01\cdot h$ in the Fourier space and in the of a perfectly symmetric grid. We compare the accuracy of the two methods:
				 the Lewis method is more accurate than the Hilbert method, when computing the CDF of the ATS, both in the case of a shift in the Fourier space (on the left) and in the symmetric (on the right). The plotted results indicate that the Hilbert method is highly unstable and even a small shift in the Fourier space  can result in a significant increase the error, up to six orders of magnitude. \label{figure:error_digital}} 
		\end{minipage}
	\end{center}

We do not desire a method that performs well either only OTM or only ITM. We want a MC that prices accurately options with any moneyness: for this reason, we consider 30 European call options with moneyness in a regular grid with range $\sqrt{t}(-0.2,0.2)$.\\   
Monte Carlo error is often decomposed into bias and variance \citep[see e.g.,][Section 1.1.3, pp.9-18]{glasserman2004monte}.
In this paper, we aim to reduce the bias error, but it is relevant to take into account also the variance. For a large number of simulations, confidence intervals estimated via MC are directly linked to this quantity \citep[see e.g.,][ch.1, eq.(1.10), p.10]{glasserman2004monte}. In our case, since we are considering the average error over 30 call options, the bias is assessed in terms of the maximum error in absolute value (MAX) wrt the exact price, while the variance is estimated with the average over the 30 MC standard deviations ($SD$). When the maximum error is below SD we can infer that the error on bias has been dealt with correctly. In all considered cases, SD is of the order of 0.1 bp and significantly above the Lewis-FFT-S error if $M> 8$. We observe such a low SD because we are using $10^7$ trials.\\

	In figure \ref{figure:error_bound}, we plot the three terms that appear in the bias bound of equation (\ref{eq:Bias_error}) for an ATS with $\alpha=2/3$ over a one-month time interval. The bound is for Lewis-FFT simulation with linear interpolation varying the number of grid points in the FFT via $M$ s.t. $N=2^M$. 
	We plot the bounds on the error due to  
	$i)$ the truncation error (blue circles) in (\ref{eq:truncation_CDF}), 
	$ii$) the linear interpolation of the CDF (red squares) in (\ref{eq:interp_CDF}), and  
	$iii$) the error related to the CDF approximation (green triangles) in  (\ref{eq:Error_FFT}). 
	As we have already anticipated in subsection \ref{subsection::error_sources}, two are the most relevant error sources: the error originating from the CDF approximation and the one due to the interpolation. The error originating from the truncation is always negligible: at least ten orders of magnitude lower than interpolation error for every $M$. 
	For the CDF approximation error, as explained in section \ref{sec:Overview}, we have suggested an optimal selection of the shift $a$ in the Lewis-FFT approach.  
	The term that we need to tackle is the interpolation one: for $M>8$ the unique relevant bound is the one on the interpolation error that scales as $\gamma^2$ for all derivative contracts with pay-off differentiable everywhere except in a finite number of points
	(e.g. for $M=10$ the interpolation error is 10 orders of magnitude above all other errors). Similar results hold $\forall \alpha \in (0,1)$.
	\begin{center}
		\begin{minipage}[t]{1\textwidth}
			\centering
			{\includegraphics[width=1\textwidth]{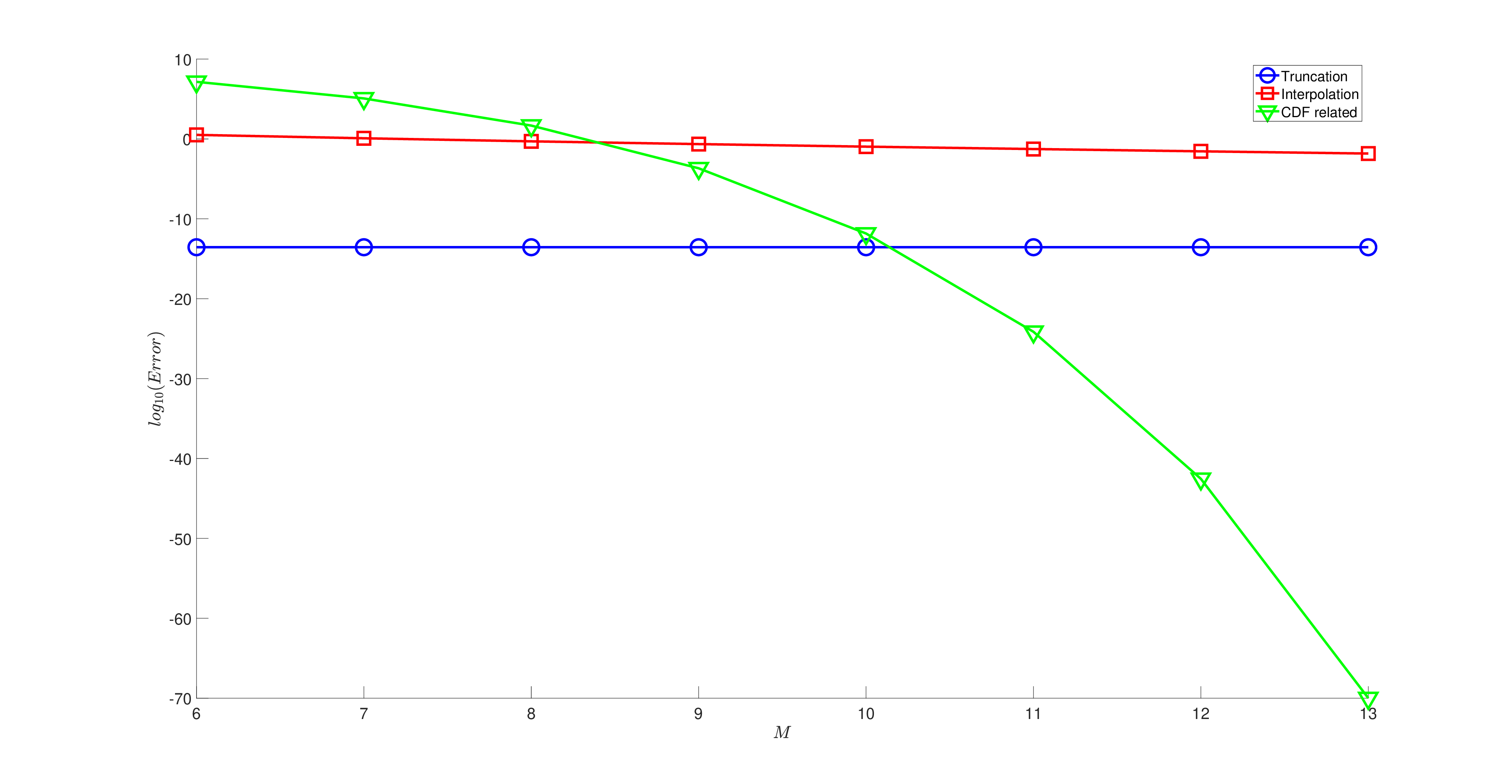} }
			\captionof{figure}{\small One-month European call option error bounds for an ATS ($\alpha=2/3$) simulated  with Lewis-FFT and linear interpolation. 
				We plot the bounds on the three error sources: 
				$i$) the truncation error \eqref{eq:truncation_CDF}  (blue circles), 
				$ii$) the error \eqref{eq:interp_CDF} due to the linear interpolation of the CDF (red squares) and 
				$iii$) the error  \eqref{eq:Error_FFT} related to numerical CDF (green triangles). Let us emphasize that the truncation error is always negligible wrt the linear interpolation error (at least $10$ orders of magnitude smaller for every $M$).
				Notice that, for $M>8$ the unique significant term is the linear interpolation error 
				(e.g. for $M=10$, it is at least $10$ orders of magnitude above all other errors). \label{figure:error_bound}} 
		\end{minipage}
	\end{center}
	As discussed in subsection \ref{subsection::error_sources}, to reduce the CDF interpolation error, 
	we consider the spline interpolation for the numerical inversion instead of the linear interpolation. 
	With spline interpolation ${\cal E}$ should scale as $\gamma^4$ instead of $\gamma^2$.
	In figure \ref{figure:1w_spline} and \ref{figure:1m_spline}, we plot the Lewis-FFT maximum error (MAX) for two different times to maturity: the error is for 30 European call options for different values of $M$ using spline (blue circles) and linear (red squares) interpolation. 
	We also plot SD, the average MC standard deviation, with a dashed green line. Notice that, for $M>6$ the spline interpolation error is significantly below the linear interpolation error. Spline interpolation's error improves significantly faster than the linear interpolation's error: for $M$ in the interval (6,10), 
	the maximum error scales as $\gamma^2$ for the linear interpolation and 
	as $\gamma^6$ for the spline interpolation.
	The observed behavior in the latter case -with an error that decreases much faster than  $\gamma^4$-
	is probably due to the monotonicity and boundness of the interpolated function (the CDF).

	\begin{center}
		
		\begin{minipage}[t]{1\textwidth}
			\centering
			{\includegraphics[width=1\textwidth]{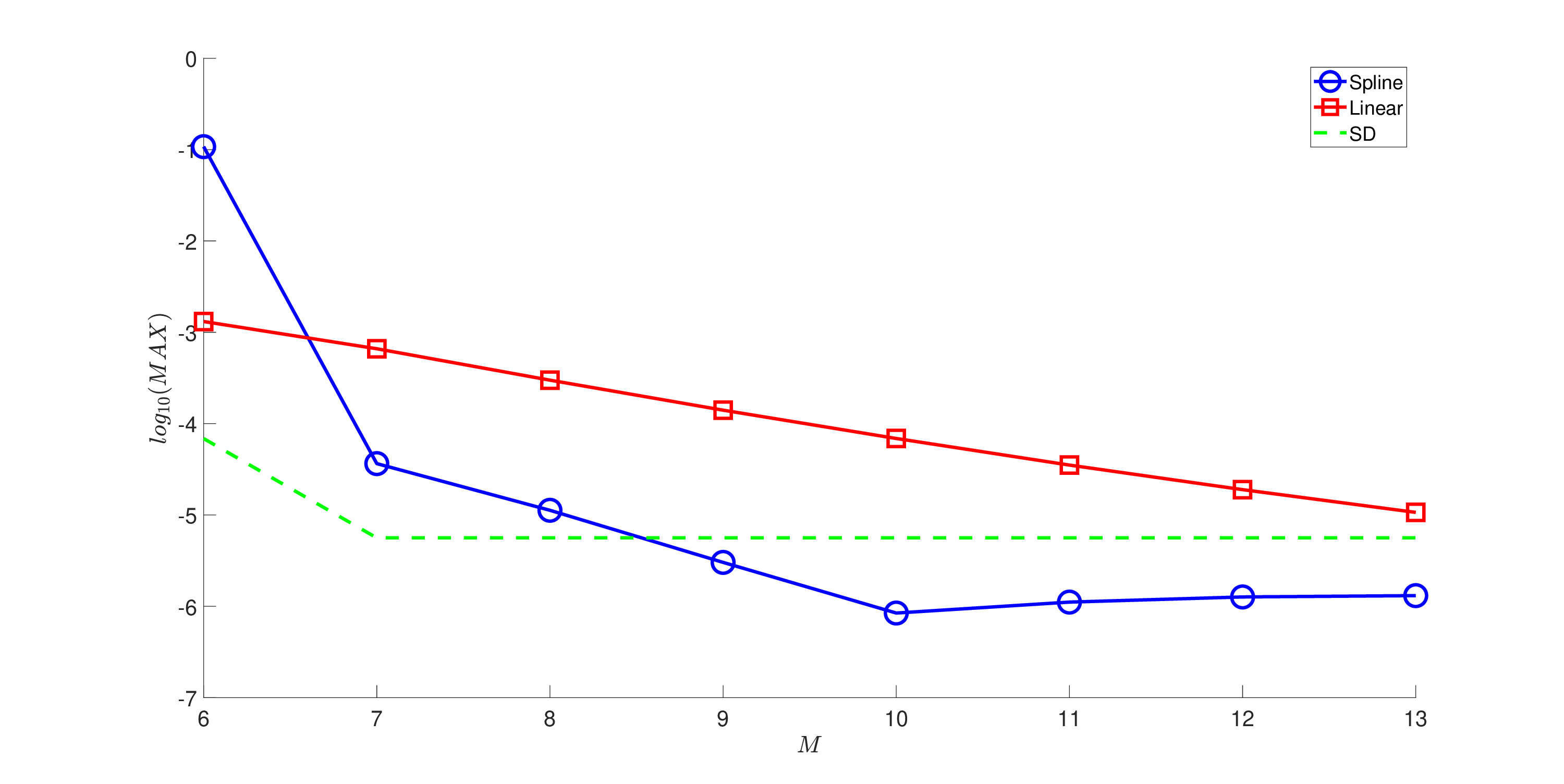} }
			\captionof{figure}{\small Maximum error for different values of $M$ using Lewis-FFT-S (blue circles) and Lewis-FFT with linear interpolation (red squares).  The maximum is computed over 30 call options (one-week maturity), with moneyness in the range $\sqrt{t}$(-0.2,0.2). We consider $10^7$ simulations and $\alpha=2/3$. Notice that, for $M>6$ the spline interpolation error is significantly below the linear interpolation error. Spline interpolation's error improves  significantly faster than the linear interpolation's error: for $M$ in the interval (6,10) the maximum error scales, on average, as $\gamma^6$ for the spline interpolation and as $\gamma^2$ for the linear interpolation.
				Moreover, the maximum error becomes significantly lower than the average MC standard deviation in a dashed green line.
				\label{figure:1w_spline}} 
		\end{minipage}
	\end{center}

	\begin{center}
		
		\begin{minipage}[t]{1\textwidth}
			\centering
			{\includegraphics[width=1\textwidth]{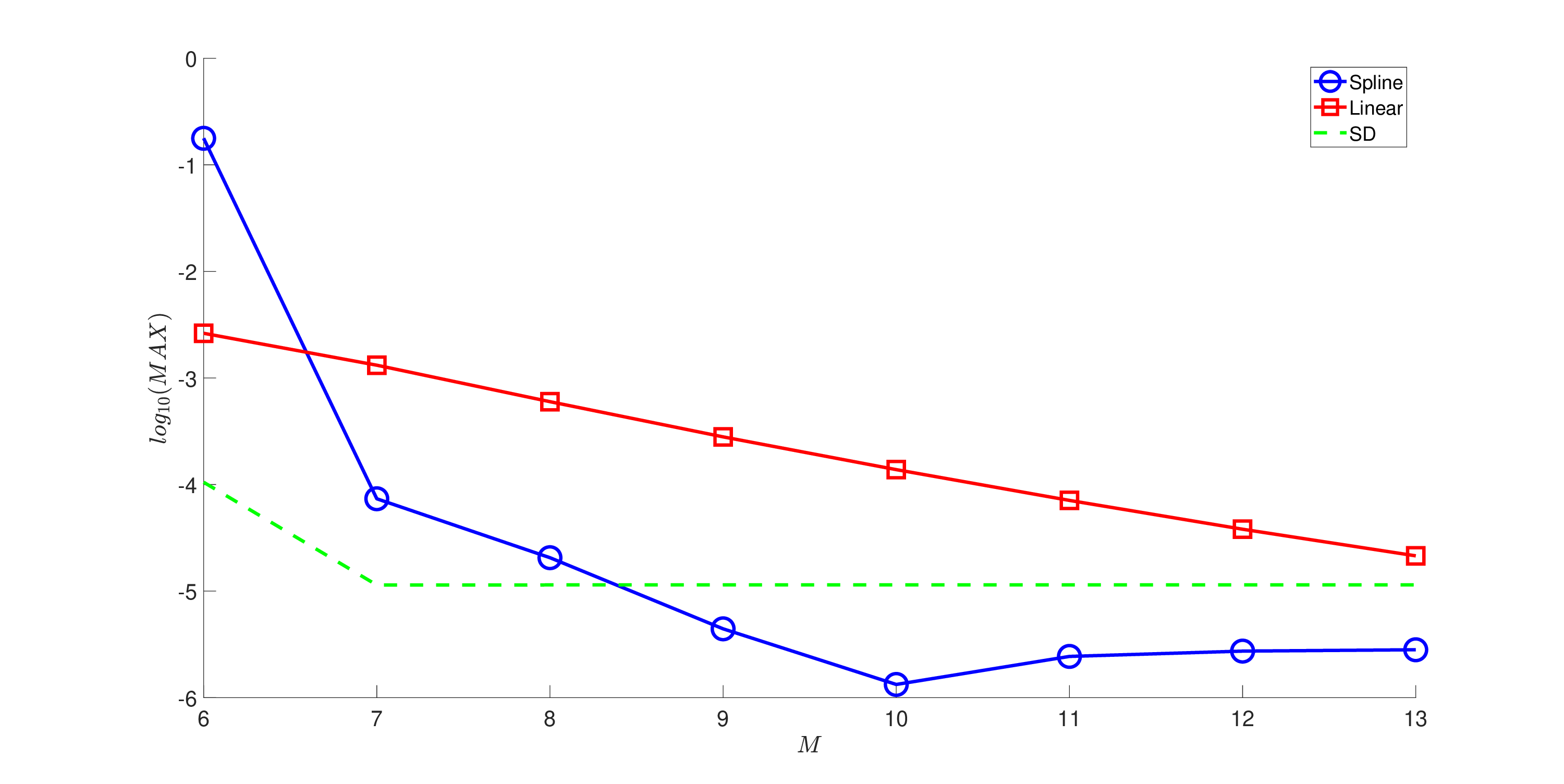} }
			\captionof{figure}{\small As figure \ref{figure:1w_spline} but for one-month maturity. Notice that, for $M>6$ the spline interpolation error is significantly below the linear interpolation error. Also in this case, spline interpolation's error improves  significantly faster than the linear interpolation's error: for $M$ in the interval (6,10) the maximum error scales, on average, as $\gamma^6$ for the spline interpolation and as $\gamma^2$ for the linear interpolation.\label{figure:1m_spline}} 
		\end{minipage}
	\end{center}

	We also desire to estimate the method's error with different metrics: besides MAX we consider the root mean squared error (RMSE) and the mean absolute percentage error (MAPE). 
	In table \ref{table::2_3_fix_N_FFT}, we report the performances of the Lewis-FFT-S algorithm  for $10^7$ simulations. We consider two values of $\alpha$ for the ATS:  $\alpha=1/3$  and $\alpha=2/3$.  The metrics are computed for 30 call options (one-month maturity) and moneyness in the range $\sqrt{t}$(-0.2,0.2). We observe that for $M> 9$ the error is 0.03 bp or below whatever metric we consider.\\
	
	\begin{table}
		\centering
		\begin{tabular}{|l|l|cccccccc|}		
			\hline
			&$\qquad \quad \;\; M$&6 &7 &8 &9 &10 &11 &12 &13\\	
			
			\hline
			
			$\alpha=1/3$	&MAX [bp.] &1639.69 &0.17 &0.02 &0.02 &0.03 &0.03 &0.03 &0.03\\
		&	RMSE [bp.] &1593.78 &0.10 &0.01 &0.01 &0.02 &0.02 &0.02 &0.02\\
		&	MAPE [\%] &1164.74 &0.07 &0.01 &0.01 &0.01 &0.01 &0.01 &0.01\\
		&	SD [bp.] &0.97 &0.12 &0.12 &0.12 &0.12 &0.12 &0.12 &0.12\\
			
			\hline	
			$\alpha=2/3$&MAX [bp.] &1774.61 &0.74 &0.21 &0.04 &0.01 &0.02 &0.03 &0.03\\
		&	RMSE [bp.] &1728.65 &0.49 &0.17 &0.03 &0.01 &0.01 &0.02 &0.02\\
		&	MAPE [\%] &1224.96 &0.34 &0.12 &0.02 &0.01 &0.01 &0.01 &0.01\\
		&	SD [bp.] &1.05 &0.11 &0.11 &0.11 &0.11 &0.11 &0.11 &0.11\\
			\hline
			
		\end{tabular}
		
		\caption{\small Lewis-FFT-S algorithm (with spline) performances wrt different metrics using $10^7$ trials for $\alpha=1/3$ and $\alpha=2/3$: MAX [bp], RMSE [bp], MAPE [\%], SD [bp]. The process is simulated for $M$ that goes from 7 to 13. The metrics are computed for 30 call options (one-month maturity), with moneyness in the range $\sqrt{t}$(-0.2,0.2). We observe that for all $M> 9$ the maximum error is 0.03 bp or below.  \label{table::2_3_fix_N_FFT}}
	\end{table}
	
	The main result of this subsection is that, in the Lewis-FFT-S framework, a Monte Carlo with $10^7$ simulations and $M>9$ provides 
	a very accurate pricing tool whatever time-horizon and $\alpha\in(0,1)$ we consider. 
	
	\subsection{European options: computational time}
	\label{subsection::Computational time}
	
	In this subsection, we emphasize that the proposed MC method is fast. We compare the Lewis-FFT-S computational cost both with the simplest possible dynamics for the underlying (geometric Brownian motion) and with the methodology that is often considered a benchmark for simulating jump processes (i.e. the simulation of jumps via the GA method). We prove that it is possible to speed up the simulation benchmark with the ziggurat method because the monotonicity conditions on $m_{s,t}^+$ and  $m_{s,t}^-$ in (\ref{eq:m})  hold for the ATS and also for the additive logistic process.\footnote{ \citet[][p.30]{EberleinMadan2009} point out that -for Sato processes- the L\'evy measure is
		decreasing in $x$ for positive $x$ and increasing in $x$ for negative $x$.}\\
	\begin{proposition}
		Consider $m_{s,t}^+$ and $m_{s,t}^-$ in (\ref{eq:m}). 	$m_{s,t}^+$ is non increasing in $x$ when $x>0$ and $m_{s,t}^-$  is non decreasing in $x$ when $x<0$ for \begin{enumerate}
			\item ATS processes with $\alpha\in(0,1)$;
			\item additive logistic processes.
		\end{enumerate}\label{pro:assumpt2}
		
	\end{proposition}
	\begin{proof}
	We have to demonstrate that $m^+_{s,t}(x)$ is non increasing in $x$ and $m^-_{s,t}(x)$ is non decreasing.
We prove the thesis by showing that the derivative of $\nu_t(x)$ wrt $x$ is 
negative and non increasing in $t$ for any $x>0$ and is positive and non decreasing in $t$ for any $x<0$.  
Notice that if this holds then 
\[
m^+_{s,t}(x)=\mathbb{I}_{x>\epsilon}\frac{\nu_t(x)-\nu_s(x)}{\int_\epsilon^\infty dz (\nu_t(z)-\nu_s(z))} 
\]
is non-increasing in $x$ and 
\[
m^-_{s,t}(x)=\mathbb{I}_{x<-\epsilon}\frac{\nu_t(x)-\nu_s(x)}{\int_{-\infty}^{-\epsilon} dz (\nu_t(z)-\nu_s(z))} 
\] 		is non decreasing in $x$.\\

We prove the thesis for the ATS.\\
Deriving $\nu_t(x)$ in (\ref{eq:nu_t}), we  get
\begin{align*}
	\frac{\partial \nu_t(x)}{\partial x}=-C_2\int_0^\infty dz \frac{e^{-z}z^{\alpha}g_3(t)\,e^{xg_2(t)}}{x^{2+\alpha}}
	\left(\alpha+\frac{z}{z/2+x\,g(t)}+1-x\,g_2(t)\right)\;\;,
\end{align*}
where $C_2$ is a positive constant.
The derivative of $\nu_t(x)$ is non increasing in $t$ for any $x>0$ because 
\begin{enumerate}
	\item $g_3(t)$ is positive and non decreasing in $t$ by condition 1 of theorem \ref{theorem:f_Additive};
	\item $e^{x\, g_2(t)}\left(\alpha +\frac{s}{s/2+x\,g(t)}\right)$ is the combination of two non decreasing function in $t$ for any $x>0$;
	\item $g_2(t)$ is negative and non decreasing and $(1-c \, x)e^{c\,x}$ is non decreasing for $c<0$.
\end{enumerate}
This proves the thesis for $x>0$.

\smallskip

The same holds true for $x<0$. \textit{Mutatis mutandis}, by substituting $g_2(t)$ with $g_1(t)$,  the proof is the same.\\

We prove the thesis for the logistic processes.\\ This entails showing that the derivative of the L\'evy measure for the CPDA model and the SLA model is non increasing for $x>0$ and non decreasing for $x<0$. Let us first consider the CPDA. Its  L\'evy measure can be rewritten as 
\begin{equation*} \displaystyle
	\nu_t(x)  = \left\lbrace\,
	\begin{array}{@{}r@{\quad}l@{}l@{}}
		\displaystyle	\frac{e^{-x/\sigma_t}}{x(1-e^{-x/\sigma_t})} &=:a_t\,g(y)   \quad \quad &x>0\\\displaystyle
		-\frac{e^{x/\sigma_t-x}}{x(1-e^{x/\sigma_t})}&=:e^{-y/a_t}a_t g(-y)   \quad \quad &x<0
	\end{array}
	\right.
\end{equation*}

where $a_t:=1/\sigma_t$, $y:=x/\sigma_t$ and $g(y):=e^{-y}/(y(1-e^{-y}))$.\\
We consider separately the positive and negative $x$. The derivative of $\nu_t(x)$ for $x>0$ is

\begin{equation*}
	\frac{\partial\nu_t(x)}{\partial x}=a_t^2\,g'(y)\;\;,
\end{equation*}
where the equality is because $\frac{\partial y}{\partial x}=a_t$.
The mixed derivative is 
\begin{align}
	\frac{\partial ^2 \nu_t(x)}{\partial x \partial t}&=a_ta_t'\left[2g'(y)+y\,g''(y)\right]=2a_ta_t'\frac{e^{y}(1+e^{y})}{(e^{y}-1)^3}<0\;\;, \label{eq:increasing_der}
\end{align}
where the first equality holds because $\frac{\partial y}{\partial t}=\frac{a'_t y}{a_t}$.\\
The derivative of $\nu_t(x)$ for $x<0$ is

\begin{align*}
	\frac{\partial \nu_t(x)}{\partial x}=-e^{-y/a_t} a_t\,g(-y)-e^{-y/a_t} a_t^2g'(-y)\;\;.
\end{align*}
We can compute the mixed derivative 
\begin{align}
	\frac{\partial^2 \nu_t(x)}{\partial x\partial t}&=
	-e^{-y/a_t} a_t'\left(g(-y)-yg'(-y)+a_t\left(2g'(-y)-yg''(-y)\right) \right)\nonumber\\ &\geq  -e^{-y/a_t}a_t'\left(g(-y)-yg'(-y)+2g'(-y)-yg''(-y)\right)=a_t'e^{-y/a_t}\frac{2e^{2y}}{(e^{y}-1)^3}>0\label{eq:decreasing_der}
\end{align}
where the  equality holds because $\frac{\partial y}{\partial t}=\frac{a'_t y}{a_t}$ and the inequalities because $a_t'<0$ and \[
2g'(y)-yg''(y)=-\frac{e^y(1+e^y)}{(e^y-1)^3}>0\;\;.
\]

Equations (\ref{eq:increasing_der}) and (\ref{eq:decreasing_der}) prove the thesis for the CPDA process.\\
The  L\'evy measure for the SLA process can be rewritten as 	\begin{equation*}
	\nu_t(x)  = \left\lbrace\,
	\begin{array}{@{}r@{\quad}l@{}l@{}}
		\displaystyle	\frac{e^{-x/\sigma_t}}{x(1-e^{-x/\sigma_t})} &=:a_t\,g(y)   \quad \quad &x>0\\\displaystyle
		-\frac{e^{x/\sigma_t}}{x(1-e^{x/\sigma_t})}&=:a_t g(-y)   \quad \quad &x<0\;\;,
	\end{array}
	\right.
\end{equation*}
Equation (\ref{eq:increasing_der}) proves the thesis for the SLA process if $x>0$. \textit{Mutatis mutandis} for the SLA process when $x<0$, we get that \[\frac{\partial ^2 v_t(x)}{\partial x \partial t}=-2a_ta_t'\frac{e^{y}(1+e^{y})}{(e^{y}-1)^3}>0\;\;.\]
Hence,  $m_{s,t}^+$ is non increasing in $x>0$ and $m_{s,t}^-$ is non decreasing in $x<0$ for both the CPDA and the SLA processes
	\end{proof}
	
	In table \ref{table::computational_cost}, we report the performances of the Lewis-FFT-S algorithm for $10^7$ simulations. 
	We consider the ATS with $\alpha=1/3$ and $\alpha=2/3$. For every choice of M, we register the computational time [s].  The metrics are computed for 30 call options (one-month maturity), with moneyness in the range $\sqrt{t}$(-0.2,0.2). We observe that for $M> 9$ the maximum error is 0.03 bp or below.\\
	
	\begin{table}
		\centering
		\begin{tabular}{|l|l|cccccccc|}
			
			\hline
			&$\qquad \;\,  M$&6 &7 &8 &9 &10 &11 &12 &13\\	
			\hline
			$\alpha=1/3$&Time [s] &0.23&0.23 &0.27 &0.28 &0.28 &0.28 &0.28 &0.29\\	
			$\alpha=2/3$&Time [s] &0.24&0.25 &0.27 &0.28 &0.28 &0.28 &0.28 &0.28\\ 
			
			\hline
		\end{tabular}
		
		\caption{\small Lewis-FFT-S computational time for simulating the ATS (with  $\alpha=1/3$ and $\alpha=2/3$) over a one-month time-interval using $10^7$ trials.  
			\label{table::computational_cost}}
	\end{table}
	
	We point out, that Lewis-FFT-S is considerably efficient. In our machine\footnote{We use MATLAB 2021a on an AMD Ryzen 7 5800H, with 3.2 GHz.}, 
	sampling  $10^7$ trials of a geometric Brownian motion takes approximately 0.08 seconds: just one-third of the Lewis-FFT-S computational cost (reported in table \ref{table::computational_cost}). 
	
	\smallskip
	
	In figure \ref{figure:time_plot}, we plot the computational time wrt the time to maturity in log-log scale for $10^7$ simulations with Gaussian Approximation (blue squares) and  Lewis-FFT-S (red circles). Time to maturity goes from one day to two years.  To compare the two methods fairly, we need to select $M$ for the Lewis-FFT-S and $\epsilon$ for the Gaussian approximation s.t. the two methods provide similar errors. As above, for both methods, we price the 30 call options, with moneyness in the range $\sqrt{t}$(-0.2,0.2). For each time to maturity, we select $M$ and $\epsilon$ s.t. the maximum error (MAX) is between 1 bp and 0.1 bp, and s.t. the Lewis-FFT-S error is always below the GA error.  Lewis-FFT-S computational time appears constant in the time to maturity, while GA computational time improves as the time to maturity reduces. However, GA is always more computationally expensive than Lewis-FFT-S by at least $1.75$ orders of magnitude. This difference appears remarkable considering that we have verified that Lewis-FFT-S error is always below GA error.   
	\begin{center}
		
		\begin{minipage}[t]{1\textwidth}
			\centering
			{\includegraphics[width=1\textwidth]{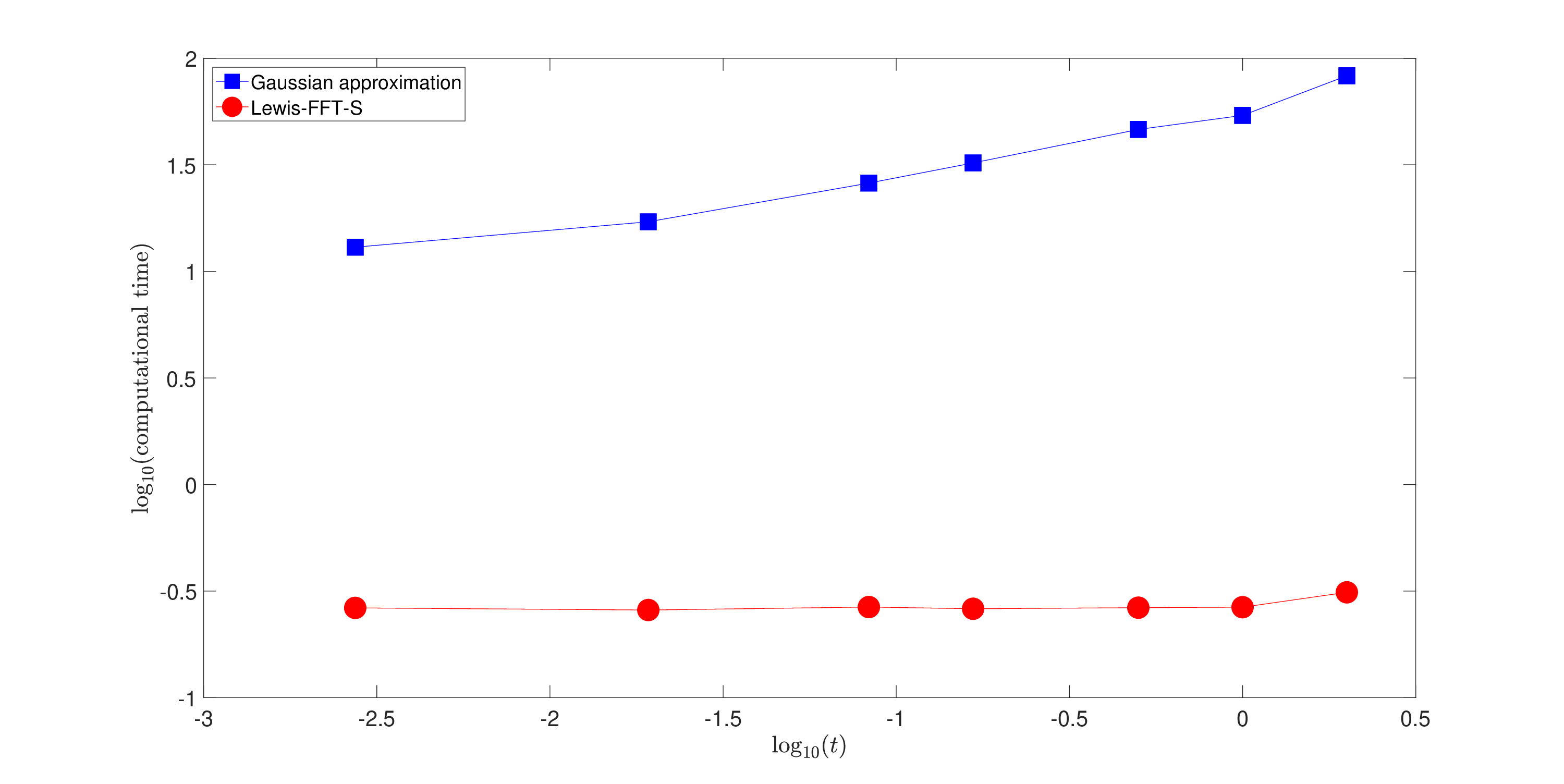} }
			\captionof{figure}{\small Computational time wrt the time to maturity in log-log scale for $10^7$ simulations with GA (blue squares) and  Lewis-FFT-S (red circles) techniques. 
				We price $30$ European call options, with moneyness in the range $\sqrt{t}  (-0.2,0.2)$ with GA and Lewis-FFT-S. 
				We consider times to maturity, between one day and two years. For each $t$, 
				we select $M$ and the threshold $\epsilon$ s.t. 
				the maximum error is between 1 bp and 0.1 bp and we require that the Lewis-FFT-S error is always below the GA error. 
				The GA computational time improves as the time to maturity reduces because a lower number of jumps is involved, while the 
				Lewis-FFT-S simulation depends weakly on the time horizon considered. We observe that GA is always more computationally expensive than Lewis-FFT-S by at least $1.75$ orders of magnitude. 
				\label{figure:time_plot}} 
		\end{minipage}
	\end{center}
	\subsection{Discretely monitoring options}
	\label{subsection::Discretly monitoring}
	In this subsection,
	to give an idea of an application of the proposed MC,
	we price discretely-monitored (quarterly) Asian options, lookback options, and Down-and-In options with a time to maturity of five years. 
	
	\medskip

	Let us call  $L$  the Down-and-In barrier.
	The payoffs -Asian calls,  lookback puts, and Down-and-In puts-  we consider are  respectively
	\begin{align*}
	&\Big(S_0\sum_{i=0}^{n}e^{f_{t_i}} -\kappa\Big)^+\;\;,\\[3mm]
	&\Big( \kappa-\min_{i}S_0e^{f_{t_i}}\Big)^+\;\;\text{and}\\[3mm]
	&\Big(\kappa-S_0e^{f_{t_i}}\Big)^+\mathbb{I}_{\min_{i}S_0e^{f_{t_i}}\geq L}\;\;,
	\end{align*}
	where $n=20$, $0=t_0<t_1<...<t_i<...<t_n$  are the monitoring times, $f_{t_i}$ is the process at time $t_i$ for the logarithm of the underlying price,  the strike price $\kappa=S_0e^{-x}$ and $x$ is the moneyness.
	For example,  the process $\{ f_t \}_{t\ge0}$ can be modeled as a Brownian motion, in the simplest Black-Scholes case, or as an ATS process, as discussed in this paper.
	In both cases, we can simulate the paths of $\{ f_t \}_{t\ge 0}$ by simulating the increments $f_{t_i}-f_{t_{i-1}}$. Every increment of the ATS is simulated separately with the Lewis-FFT-S method. We point out that the procedure can be parallelized by leveraging on the independence of increments.
	
	\medskip

	In table \ref{table::Esotic}, we report prices and MC standard deviation of Asian calls, lookback puts, and Down-and-In puts (with a barrier strike $L=0.6$). 
	We simulate $10^7$ paths of the ATS with $\alpha=2/3$ and price the discretely-monitored (quarterly) path-dependent options with time to maturity of five years. 
	We consider options with different moneyness in the range (-0.5,0.5), where $0.5\approx 0.2\sqrt{t}$ for $t=5$ years. We use $M=13$ for the numerical CDF inversion. 
	The method is very precise: the numerical error SD is of the order of one bp (or below) in all considered cases.
	
	As pointed out in the previous subsection, the Lewis-FFT-S is also extremely efficient when pricing discretely-monitored path-dependent exotics: 
	with an ATS, it takes only three times the computational cost 
	that it takes with a standard  Brownian motion.
	\begin{table}
		\centering

		\begin{tabular}{|c|cc|cc|cc|}
			\hline
			Moneyness &Asian [\%] &SD  [\%] &Lookback [\%] &SD  [\%] &Down-and-In [\%] &SD  [\%]\\

			\hline
		-0.5 &39.79 &0.01 &3.31 &0.00 &2.31 &0.00\\
		-0.25 &24.36 &0.01 &8.72 &0.00 &3.98 &0.00\\
		0 &10.04 &0.01 &23.07 &0.01 &6.15 &0.01\\
		0.25 &2.57 &0.00 &50.53 &0.01 &8.95 &0.01\\
		0.5 &0.55 &0.00 &86.98 &0.01 &12.55 &0.01\\
			
			\hline
		\end{tabular}
		
		\caption{\small Prices and MC standard deviation of Asian calls, lookback puts, and Down-and-In puts for  moneyness in the range (-0.5,0.5). 
			We simulate $10^7$ paths of the ATS with $\alpha=2/3$ and price the discretely-monitored (quarterly) path-dependent options with time to maturity of five years. SD errors are always lower than 1 bp. \label{table::Esotic}}
	\end{table}

	\section{Conclusions}
	\label{sec:Conclusion}
	
	In this paper, we propose the Lewis-FFT-S method: a new Monte Carlo scheme for additive processes that leverages on the numerical efficiency of the FFT 
	applied to the Lewis formula for a CDF 
	and on the spline interpolation properties when inverting the CDF. We present an application to the additive normal tempered stable process, which has excellent calibration features on the equity volatility surface \citep[see e.g.,][]{azzone2019additive}.
	This simulation scheme is accurate and fast.\\
	
	We discuss in detail the accuracy of the method. 	
	In figure \ref{figure:error_bound}, we analyze the three-components of the bias error in (\ref{eq:Bias_error}). 
	In this study, we have shown how to accelerate convergence 
	by improving the two main sources of numerical error (\ref{eq:Bias_error}): 
	the  CDF error  (\ref{eq:Error_FFT}) and the interpolation error (\ref{eq:interp_CDF}). 
	First, we sharpen the  CDF error considering the Lewis formula  (\ref{eq:Lewis CDF}) for CDF and selecting the optimal shift that minimizes the error bound in the FFT. 
	Second, we substitute the linear interpolation with the spline interpolation. In this way, the leading term in the interpolation error improves from $\gamma^2$ to at least  $\gamma^4$.
	This improvement is particularly evident in figures \ref{figure:1w_spline}-\ref{figure:1m_spline}, where, for $M>6$, 
	the Lewis-FFT-S maximum error is significantly below the Lewis-FFT version of the method with linear interpolation and it appears to decrease as $\gamma^6$ in numerical experiments.\\
	
	The Lewis-FFT-S is also fast. As discussed in subsection \ref{subsection::error_sources}, for a sufficiently large number of simulations, the increment in computational time due to spline interpolation is negligible.
	Moreover, as shown in figure \ref{figure:time_plot}, the proposed method is at least one and a half orders of magnitude faster than the traditional GA simulations, whatever time horizon we consider.
	Finally, we observe that, when pricing some  discretely-monitoring path-dependent options, the computational
	time is of the same order of magnitude as standard algorithms for Brownian motions.

	\section*{Acknowledgements}
	We are grateful for the valuable comments from  M. Fukasawa and C. Alasseur. We thank L. Ballotta, F. Baschetti, J. Blomvall, G. Bormetti, G. Consigli, G. Fusai, G. Germano, I. Kyriakou, O. Le Courtois, A. Pallavicini, R. Renò, and all participants to the conference on Intrinsic Time in Finance in Constance, to the QF workshop 2022 in Rome, to the ECSO-CMS conference 2022 in Venice, to the seminar at EDF R\&D in Paris and to the Bayes financial engineering workshops 2022 in London. 
	\newpage
	\bibliography{sources}
	\bibliographystyle{tandfx}
	\clearpage
	\section*{Notation and shorthands}
	
	\begin{center}

		\begin{tabular} {|c|l|}
			\toprule
			\textbf{Symbol}& \textbf{Description}\\ \bottomrule
			$\beta$& scaling parameter of ATS variance of jumps\\
			$\delta$ &scaling parameter of ATS skew parameter\\
			${\cal E}$ & total error when pricing the derivative with payoff $V(x)$ \\ 
			$\eta_t$& ATS skew parameter\\
			$\bar{\eta}$& ATS constant part of skew parameter\\
			$\epsilon$  & small jump threshold for GA\\
			${\cal E}^{CDF}_{h, M}$ & CDF error bound as a function of the grid size $h$ and of $M$ \\ 
			${\cal E}^{CDF}_{ M}$ & CDF error bound when $h$ is s.t. the two sources of CDF error are comparable \\ 
			$f_t$ & the additive process at time $t$\\
			$F_t(T)$ &forward price with maturity $T$ at time $t$\\
			$\phi_{s,t}$& characteristic function of additive increment between $s$ and $t$ time to maturity\\
			$\phi_{t}$& characteristic function of additive process at time $t$\\
			$\gamma$ & grid size in the CDF domain\\
			$h$ &  grid size in the Fourier domain \\
			$k_t$& ATS variance of jumps parameter\\
			$\bar{k}$& ATS constant part of the variance of jumps\\
			$K$& dimension of the CDF interpolation grid\\
			$\kappa$& strike price\\
			$L$ & Down-and-In barrier strike\\
			$m^+_{s,t}$ & probability density of positive jumps\\
			$m^-_{s,t}$&		probability density of negative jumps\\
			$M$ & integer number s.t. ${ N }$ is the number of grid points\\
			$n$ & number of monitoring times in path dependent derivatives\\
			$n_v$& number of points in which $V$ is not differentiable\\
			$N$ &  number of grid points  $(N={ 2^M })$ \\
			${\cal N}_{sim}$ & number of simulations \\ 
			$\nu_t$ &jump measure of additive process\\
			$P(x)$ & model CDF of the increment between the times $s$ and $t$ \\
			$\hat{P}(x)$ & numerical approximation of the CDF of the increment between the times $s$ and $t$ \\
			$p^-_t$ & upper bound of $\phi_{t}$ strip of regularity\\
			$p^+_t$ & $-(p^+_t+1)$ is the lower bound of $\phi_{t}$ strip of regularity\\ 
			$\bar{\sigma}$& ATS diffusion parameter\\
			$S_t$& spot price at time t\\
			$U$ & uniform r.v. in (0,1)\\
			$V(x)$ & derivative payoff\\
			$(x_0,x_K)$& interval in which the CDF is interpolated\\

			\bottomrule
			
		\end{tabular}
		
	\end{center}	
	
	\newpage
	{\bf Shorthands}
	
	\begin{center}
		\begin{tabular}{|c|l|}
			\toprule
			\textbf{Symbol}& \textbf{Description}\\ \bottomrule

			a.s. & almost surely\\
			ATS &additive normal tempered stable process\\
			bp& basis point\\
			CDF & cumulative distribution function \\
			FFT & fast Fourier transform \\
			GA & Gaussian approximation technique \\
						MAE & mean absolute error\\
			MAPE & mean absolute percentage error (in MC prices)\\
			MAX & maximum error  (in MC prices)\\
			MC & Monte Carlo \\
			ms & milliseconds\\
			nn & nearest neighborhood algorithm \\
			r.v. & random variable \\
			RMSE & MC prices root mean squared errors \\
			SD & average standard deviation  (in MC prices)\\
			s.t & such that\\
			wrt & with respect to\\
			\bottomrule
		\end{tabular}
	\end{center}
	
	\begin{appendices}
		\section{The key features of the ATS process}

	\label{appendix_ATS}
In this appendix, we briefly recall the features of the ATS process that we use in the numerical experiments.		

As in \citet{azzone2019additive}, we model the forward at time $t$ with maturity $T$ as an exponential additive
\[F_t(T)=F_0(T)e^{f_t}\;\;,\]
where $f_t$ is the ATS process.

		At time $t$, the ATS characteristic function is
		\begin{equation}
		\label{laplace}
		\phi_t(u)=\E \, e^{i \; u \; f_t }={\cal L}_t \left(iu \left(\frac{1}{2}+\eta_t \right)\sigma_t^2+\frac{u^2\sigma^2_t}{2};\;k_t,\;\alpha \right)e^{-iu\log{\cal L}_t\left(\eta_t\sigma^2_t;\;k_t,\;\alpha\right)}\;\; .
		\end{equation}
		$ \sigma_t $, $k_t$  are continuous on $[0,\infty)$ and $ \eta_t $ is continuous  on $(0,\infty)$,
		with $ \sigma_t > 0$, $ k_t, \eta_t \geq  0$.  As in the corresponding  L\'evy case, we define   
		\[ 
		\ln {\cal L}_t \left(u;\;k,\;\alpha\right) :=
		\displaystyle \frac{t}{k}
		\displaystyle \frac{1-\alpha}{\alpha}
		\left \{1-		\left(1+\frac{u \; k}{1-\alpha}\right)^\alpha \right \}  \;\; ,
		\]
		with $ \alpha \in (0,1) $.\footnote{
			We emphasize that we consider $\alpha >0$. 
			As discussed in subsection \ref{subsection::error_CDF}, this is the relevant situation in practice  when pricing  exotic derivatives:
			the case with $\alpha$ exactly equal to zero presents a power-law decay in the characteristic function. 
		} 
	
	As proven by \citet{azzone2019additive} in proposition 2.2, the forward process $F_t(T)$ is a martingale under the risk neutral measure.
		
		The ATS  jump measure  is
		\begin{equation}
		\nu_t(x)=\dfrac{tC\left(\alpha,k_t,{\sigma}_t,\eta_t\right)}{|x|^{1/2+\alpha}}e^{-(1/2+\eta_t)x}K_{\alpha+1/2}\left(|x|\sqrt{{\left(1/2+\eta_t\right)^2+2(1-\alpha)/( k_t\,{\sigma}^2_t)}}\right)	
		\;\;, \label{eq:nu_t}
		\end{equation}
		with \[C\left(\alpha,k_t,{\sigma}_t,\eta_t\right):=\frac{2}{\Gamma(1-\alpha) \sqrt{2 \pi}}\left(\frac{1-\alpha}{k_t}\right)^{1-\alpha}{\sigma}^{2\alpha}_t\left(\left(1/2+\eta_t\right)^2+2(1-\alpha)/(k_t\, {\sigma}^2_t)\right)^{\alpha/2+1/4}\;\;,\]
		and	$ K_\nu(x)$ the modified Bessel function of the second kind  \citep[see e.g.,][ch.9, p.376]{abramowitz1948handbook}
		\[ 
		K_\nu(x):=\frac{e^{-x}}{\Gamma\left(\nu+\frac{1}{2}\right)}\sqrt{\frac{\pi}{2 \, x}}\int_0^\infty dz e^{-z}z^{\nu-1/2}\left(1+\frac{z}{2 \, x}\right)^{\nu-1/2}\;\;.
		\]

		Moreover,	we recall that a	sufficient condition for the existence of ATS is provided in the following theorem \citep[cf.][th.2.1, p.503]{azzone2019additive}.
		\begin{theorem} {\bf Sufficient conditions for existence of ATS} \label{theorem:f_Additive}\\ 
			There exists an additive process $ \left\{f_t \right\}_{t\geq 0}$ with the characteristic function  (\ref{laplace}) if the following two conditions hold.
			\begin{enumerate}[(a)]
				\item $g_1(t)$, $g_2(t)$, and $g_3(t)$ are non decreasing, where  
				\begin{align}
				g_1(t)&:=(1/2+\eta_t)-\sqrt{\left(1/2+\eta_t\right)^2+ 2(1-\alpha)/(\sigma_t^2 k_t)} \nonumber\\
				g_2(t)&:=-(1/2+\eta_t)-\sqrt{\left(1/2+\eta_t\right)^2+ 2(1-\alpha)/( \sigma_t^2 k_t)}\label{eq:g_2}\\
				g_3(t)&:=\frac{ t^{1/\alpha}\sigma^2_t}{k_t^{(1-\alpha)/\alpha}}\sqrt{\left(1/2+\eta_t\right)^2+ 2(1-\alpha)/( \sigma_t^2k_t)}\;\;; \nonumber
				\end{align}
				\item Both  $t\,\sigma_t^2\,\eta_t$ and $t\,\sigma_t^{2\alpha}\,\eta_t^\alpha\,/k_t^{1-\alpha}$ go to zero as $t$ goes to zero  \qed
			\end{enumerate} 
			
		\end{theorem}
	
	We point out that the boundaries of the strip of regularity of the characteristic function of the ATS $p_t^++1$ and $p_t^-$ are equivalent to $g_1(t)$ and $g_2(t)$, as shown in the proof of proposition \ref{proposition::assumptions}.

		

		\vspace{2cm}		
		\
		\newpage
		\section{Simulation algorithm} \label{appendix:simulation_algo}
		A brief description of the Lewis-FFT algorithm follows\\

		\begin{algorithm}
			
			\begin{algorithmic}
				
				\Procedure{Lewis-FFT}{$M,{\cal N}_{sim}$,$FlagSpline$}       
				\State
				\State COMPUTE $h(M)$, $N$, $\gamma$
				\State COMPUTE $ \hat{P}$ with FFT						\Comment{$z_0,\;z_{N-1}$ fixed by FFT}
				\State
				\State FIX $x_K$ nearest point to $5\sqrt{t-s}$ and $x_0=-x_K$
				\State $\vv{x}=x_0:\gamma:x_K$  \Comment{Grid dimension: $K+1$}
				\State
				\State SAMPLE a vector ${U}$ of ${\cal N}_{sim}$ uniform r.v. in [0,1]
				\State ${J}$ = NearestNeighborhood(${U}$, $\hat{P}(\vv{x})$)   \Comment{Find next element in the grid}
				\State
				\If{$FlagSpline = True$}
				\State COMPUTE spline interpolation coefficients $\{c_{q,J}^S\}_{q=0}^3$  \Comment{Solve tridiagonal linear sistem}
				\State ${\rm X}={\rm spline}(\hat{P}(\vv{x})$, $\vv{x}$, $U$,$J$) 		\Comment{Interpolate on $U$}
				\Else
				\State  COMPUTE linear interpolation coefficients $\{c_{q,J}^L\}_{q=0}^1$
				\State ${\rm X} =c_{0,J}^L +{U}\,c_{1,J}^L$					\Comment{Interpolate on $U$}
				
				\EndIf
				\EndProcedure
				
			\end{algorithmic}
		\end{algorithm}
		
		\section{European options errors: Spline vs Linear}
		In this appendix, we report the error between simulated and exact option prices strike-by-strike.
		In table \ref{table::EuropeanOptionPrices}, we report the prices of 30 European options with exact method, Lewis-FFT-S MC and Lewis-FFT MC with linear interpolation (in percentage) for the 1-month maturity, $\alpha=2/3$, and $M=10$. Option prices are in percentage of the spot price. Errors of the Lewis-FFT-S are of the order of 0.01 bp. Morover, errors with spline interpolation are, on average, two orders of magnitude below errors with linear interpolation.
			\begin{table}
			\centering
	\resizebox{\textwidth}{!}{	\begin{tabular}{|ccc|ccc|ccc|}
		\hline
Strike [\%] &x [\%]&Exact [\%]&Lewis-FFT-S [\%]&Error [bp]& Rel. Error [\%]&Lewis-FFT-Lin [\%] &Error [bp]& Rel. Error [\%]\\
\hline
105.94 &-5.77 &0.42 &0.42 &-0.01 &-0.03 &0.43 &0.54 &1.28\\
105.52 &-5.38 &0.48 &0.48 &-0.01 &-0.02 &0.48 &0.65 &1.37\\
105.10 &-4.98 &0.54 &0.54 &-0.00 &-0.01 &0.55 &0.70 &1.30\\
104.69 &-4.58 &0.61 &0.61 &0.00 &0.00 &0.62 &0.77 &1.27\\
104.27 &-4.18 &0.69 &0.69 &0.01 &0.01 &0.70 &0.90 &1.31\\
103.86 &-3.78 &0.77 &0.77 &0.01 &0.01 &0.78 &1.00 &1.29\\
103.44 &-3.38 &0.87 &0.87 &0.01 &0.01 &0.88 &1.03 &1.18\\
103.03 &-2.99 &0.98 &0.98 &0.01 &0.01 &0.99 &1.11 &1.13\\
102.62 &-2.59 &1.10 &1.10 &0.01 &0.01 &1.11 &1.21 &1.10\\
102.21 &-2.19 &1.22 &1.23 &0.01 &0.01 &1.24 &1.26 &1.03\\
101.81 &-1.79 &1.37 &1.37 &0.01 &0.01 &1.38 &1.29 &0.94\\
101.40 &-1.39 &1.52 &1.52 &0.01 &0.01 &1.53 &1.33 &0.88\\
101.00 &-1.00 &1.69 &1.69 &0.01 &0.01 &1.70 &1.37 &0.81\\
100.60 &-0.60 &1.87 &1.87 &0.01 &0.01 &1.88 &1.38 &0.74\\
100.20 &-0.20 &2.06 &2.06 &0.01 &0.01 &2.07 &1.37 &0.66\\
99.80 &0.20 &2.26 &2.26 &0.01 &0.00 &2.27 &1.34 &0.59\\
99.40 &0.60 &2.48 &2.48 &0.01 &0.00 &2.49 &1.33 &0.54\\
99.01 &1.00 &2.71 &2.71 &0.01 &0.00 &2.72 &1.31 &0.48\\
98.62 &1.39 &2.95 &2.95 &0.01 &0.00 &2.96 &1.24 &0.42\\
98.22 &1.79 &3.20 &3.20 &0.00 &0.00 &3.21 &1.16 &0.36\\
97.83 &2.19 &3.46 &3.46 &-0.00 &-0.00 &3.47 &1.12 &0.32\\
97.45 &2.59 &3.73 &3.73 &-0.01 &-0.00 &3.74 &1.06 &0.28\\
97.06 &2.99 &4.01 &4.01 &-0.01 &-0.00 &4.02 &0.95 &0.24\\
96.67 &3.38 &4.29 &4.29 &-0.01 &-0.00 &4.30 &0.87 &0.20\\
96.29 &3.78 &4.59 &4.59 &-0.01 &-0.00 &4.60 &0.85 &0.19\\
95.91 &4.18 &4.89 &4.89 &-0.00 &-0.00 &4.90 &0.77 &0.16\\
95.52 &4.58 &5.20 &5.20 &-0.00 &-0.00 &5.21 &0.68 &0.13\\
95.14 &4.98 &5.51 &5.51 &0.00 &0.00 &5.52 &0.63 &0.11\\
94.77 &5.38 &5.83 &5.83 &0.00 &0.00 &5.84 &0.60 &0.10\\
94.39 &5.77 &6.15 &6.15 &0.00 &0.00 &6.16 &0.52 &0.09\\
		\hline
	\end{tabular}}

			\caption{\small Prices of 30 European options with exact method, Lewis-FFT-S MC and Lewis-FFT MC with linear interpolation (in percentage) for the 1-month maturity, $\alpha=2/3$, and $M=10$. Option prices are in percentage of the spot price. For every option we compute the error (exact price - MC price) in bp and the relative error  (exact price - MC price) divided by the exact price in percentage.  Errors of the Lewis-FFT-S are of the order of 0.01 bp. Morover, errors with spline interpolation are, on average, two orders of magnitude below errors with linear interpolation. \label{table::EuropeanOptionPrices}}
		\end{table}
	\label{appendix:table}
		
	\end{appendices}
\end{document}